\documentclass{article}
\usepackage{times}
\usepackage{courier}
\usepackage{amsmath}
\usepackage{amsthm}
\usepackage{amssymb}
\usepackage{graphicx}
\usepackage{multirow}
\usepackage[usenames]{color}
\usepackage{algorithm}
\usepackage{algpseudocode}
\usepackage{enumerate}
\usepackage{float}
\usepackage{geometry}
\usepackage{setspace}

\newcommand{\winner}{{\mathcal F}}
\newcommand{\pSet}{{\mathcal L}(C)}
\newcommand{\vecQ}{{\mathbf Q}}
\newcommand{\famQ}{{\mathcal Q}}
\newcommand{\vecP}{{\mathbf P}}
\newcommand{\score}{{\displaystyle{score}}}
\newcommand{\segment}[3]{[#1\ ;#3\ ;#2]}
\newcommand{\segleft}[3]{(#1\ ;#3\ ;#2]}
\newcommand{\segright}[3]{[#1\ ;#3\ ;#2)}
\newcommand{\segopen}[3]{(#1\ ;#3\ ;#2)}

\newcommand{\kendal}{{d_{swap}}}

\newtheorem{definition}{Definition}

\newtheorem{theorem}{Theorem}
\newtheorem{cor}{Corollary}
\newtheorem{lemma}{Lemma}

\providecommand{\keywords}[1]{\textbf{\textit{Keywords:}} #1}
 
\title{Lie on the Fly:\\ Strategic Voting in an Iterative Preference Elicitation Process  }
\author{Lihi Naamani-Dery$^1$, Svetlana Obraztsova$^2$, Zinovi Rabinovich$^3$, Meir Kalech$^4$
\\ $^1$Ariel University, Ariel, Israel. lihid@ariel.ac.il
\\ $^2$Nanyang Technological University, Singapore. lana@ntu.edu.sg
\\ $^3$Nanyang Technological University, Singapore. zinovi@ntu.edu.sg
\\ $^4$Ben-Gurion University, Beer Sheva, Israel. kalech@bgu.ac.il
}
\date{}

\begin{document}
\maketitle

\begin{abstract}

A voting center is in charge of collecting and aggregating voter preferences. In an iterative process, the center sends comparison queries to voters, requesting them to submit their preference between two items. 
Voters might discuss the candidates among themselves, figuring out during the elicitation process which candidates stand a chance of winning and which do not. Consequently, strategic voters might attempt to manipulate by deviating from their true preferences and instead submit a different response in order to attempt to maximize their profit. We provide a practical algorithm for strategic voters which computes the best manipulative vote and maximizes the voter's selfish outcome when such a vote exists. We also provide a careful voting center which is aware of the possible manipulations and avoids manipulative queries when possible. In an empirical study on four real world domains, we show that in practice manipulation occurs in a low percentage of settings and has a low impact on the final outcome. 
The careful voting center reduces manipulation even further, thus allowing for a non-distorted group decision process to take place.
We thus provide a core technology study of a voting process that can be adopted in opinion or information aggregation systems and in crowdsourcing applications, e.g., peer grading in Massive Open Online Courses (MOOCs). 
\end{abstract}

\section{Introduction}
Voting procedures are used for combining voters' individual preferences over a set of alternatives, enabling them to reach a joint decision. However, sometimes the full set of preferences is unavailable.
Take, for example, a recruiting committee that convenes to decide on the appropriate candidate to fill a position. Ideally, each committee member is required to rank all applicants; then a joint decision is reached based on all opinions (see e.g. \cite{chen2013mining}). However, as their time is limited, each committee member is reluctant to describe and disclose a complete list of ranked preferences (see e.g. the discussion in \cite{xia2011maximum}). As another example, consider peer grading in Massive Open Online Courses (MOOCs). Since students are not professional educators, they are not trained to provide grades in absolute terms. Rather, students provide comparative information by answering some binary comparative queries (see e.g. \cite{capuano_etal_MOOC_2017}). Even when the voter is acquainted with all of the candidates, it is easier to answer relative comparison queries than to rank all of the alternatives \cite{balakrishnan2012two}. Furthermore, voters are more accurate when making relative comparisons than (a) ranking all items \cite{miller1956magical}; and than (b) presenting precise numerical values \cite{damart2007supporting}.

To minimize the degree of open preference formulation, a preference elicitation technique can be applied. One such technique is iterative relative comparisons, i.e., querying the voters for their preference between two candidates. 
It has been theoretically shown that not all preferences are needed for reaching a joint decision when full information about voter preferences is available, and that the same joint decision can be reached with partial information \cite{Konczak05}. Consequently some practical algorithms for eliciting a minimal set of user preferences using relative comparisons are available, showing that in practice the required information can be cut by more than 50\% \cite{lu_boutilier_2011, naamani2015}. 

These algorithms assume that voters are sincere in their response and that intermediate results are not available to the voters. However these assumptions might not hold in real world scenarios. As a toy example, consider three candidates: $c_1,c_2,c_3$. After a few iterations of voting is it apparent that either $c_2$ or $c_3$ will be elected (since $c_1$ received zero votes and almost all of the voters have been queried). Now, assume that voter $v_1$ prefers $c_1$ over $c_2$. However, given that $c_1$ stands no chance of winning, $v_1$ may choose to vote strategically and state that he prefers $c_2$ over $c_1$ in order to reduce $c_3$'s chance of winning.
To the best of our knowledge, iterative preference elicitation has not been studied with strategic behavior nor has the issue of intermediate results been addressed. Specifically we assume that:

\begin{itemize}
\item \textbf{Intermediate results are available} - it is possible that the voters learn the intermediate results either directly, e.g., if the intermediate results are published during the process; or indirectly, e.g., the voters discuss their preferences and reveal how they responded to the relative comparison queries they received; or through an information leak, e.g., someone tapped in on the voting center in charge of aggregating the voters.
\item \textbf{Strategic behavior is possible} - voters may adopt a strategic behavior in order to manipulate the outcome of the election towards their favorite candidates ~\cite{Far69,laf87}. This behavior has been observed in online applications (see e.g.~\cite{zou_meir_parkes_2015}).
\end{itemize}

In this paper, we set out to study strategic behavior in incremental voting processes.
We follow \cite{lu_boutilier_2011} and \cite{naamani2015} and assume that the voting center proceeds in rounds; in each round the voting center selects one relative comparison query, i.e., one voter to query for her preference between two candidates.
Deviating from previous research, we assume that at the end of each round, the voting center directly exposes the intermediate results, i.e., the candidates that stand a chance to win the elections. Thus the voters may attempt to manipulate by sending the voting center an insincere response in order to promote or avoid certain candidates according to their truthful preferences.
Note that voters cannot erase their past responses. In fact, the center never queries when it can infer an answer, as such a query would be wasted in terms of eliciting new information. Essentially, the center maintains a closure by inference of all obtained preferences. This means that the center would immediately recognize query responses inconsistent with previously obtained preferences. As with other knowledge bases, such inconsistencies can be resolved  \cite{brewka_91,brewka_etal_97}, the simplest resolution being to invalidate the voter or even election entirely. As a result, the very process of finding an effective strategic vote becomes a non-trivial problem.

\textbf{Contributions:}
There are two main challenges in a model that allows strategic voting in iterative preference elicitation. One challenge is from the voting center's point of view, the second challenge is from the voter's point of view.
The voting center's main challenge is to try and avoid manipulations by presenting only ``safe" relative comparison queries to the voters, i.e., queries in which the voters must answer truthfully since they are not able to compute a strategic vote.
Nevertheless, each voter's challenge is to maximize her own selfish outcome, as is done in the voting literature since the classical works of \cite{Gib73} and \cite{Sat75}. This is achieved if and when the voters are able to compute a strategic vote.
For the first challenge, we present a \textbf{careful voting center} (Section~\ref{sec:ild:sub:careful_voting_center}) which tries to prevent manipulation. The voting center will not present queries over candidates that can be manipulated by a strategic voter, unless it has no other choice. However, this is only a conclusion to our development of a wide range of possible manipulations. 
More specifically, for the second challenge, we provide a \textbf{manipulation algorithm for voters} (Section~\ref{sec:ild:sub:algo_build}) in an iterative voting setting. We present a dominance-based heuristic in order for voters to submit strategic responses that will maximize their outcome and will not contradict their previous statements.
We prove the soundness and completeness of the manipulation algorithm, meaning, if a manipulative response (a preference over queried candidates) exists then a manipulative voter will submit it to the voting center.
By addressing these challenges we provide the foundations for a realistic group decision support system.

 We evaluate the impact of the careful voting center and the manipulation algorithm on four real-world data sets.
 We compare manipulative voters to truthful voters in (a) a careful voting center and (b) a na{\"i}ve voting center that does not attempt to avoid manipulations.
 We examine three research questions: (i) How often do manipulations occur in practice? (ii) Is the final result impacted by manipulations? (iii) How do manipulations impact the number of iterations required to reach the final result? We show that in practice, manipulation happens in a low percentage of the settings and has a low impact on the final outcome. A careful voting center reduces the occurrence of manipulation even further,
 but the tradeoff is that in turn, there is an increase in the amount of queries needed in order to end the iterative process and find a winner.

 Since iterative preference elicitation with strategic behavior is a novel idea, we chose to begin with the Borda voting rule \cite{de1781memoire}.
 Borda voting is considered a robust and accurate voting rule. The Borda rule is a positional voting rule whose scores are generated by ordered weighted averaging \cite{Lapresta_2017}.
 It must be noted that several general results have been obtained regarding the complexity of manipulating the Borda voting rule with partial information \cite{conitzer_walsh_xia_2011,Conitzer11}. However, these works assume that a complete preference order is submitted as a ballot. Moreover, it has been shown that for Borda, partial voting increases the situations where strategic voting is possible, and the complexity of computing a manipulation decreases \cite{narodytska2014computational}. This, in turn, supports the need to construct a careful voting center.

An earlier and shorter version of this work was published in the proceedings of IJCAI-2015 \cite{Dery15}. In this paper we expand both the theoretical and empirical parts of the research. Beyond providing detailed explanations of the manipulation algorithm for voters in an iterative voting setting, we formalize the constraints in which a voter is able to manipulate the voting. Also, we add a proof to the soundness and completeness of the manipulation algorithm. An additional important contribution is in the wide set of experiments added to this version. We added three real-world data sets in order to prove that our conclusions are domain independent. The new experiments analyze interesting aspects of the manipulation algorithms and examine its impact in practice.

This paper is structured as follows. In Section~\ref{sec:rel_work} we survey the current work on preference elicitation with partial information and the current work on iterative voting processes.
In Section~\ref{Model} we lay out the model preliminaries before we present a detailed algorithm in Section~\ref{ivc_ldm_theory}. In Section~\ref{sec:exps} we set forth the empirical study carried out to answer our research questions. Finally, we provide some conclusions in Section~\ref{sec:conclude}.

\section{Related Work}\label{sec:rel_work}
In this paper we fuse together two distinct subjects of interest: a) multi-query (iterated) preference elicitation and b) iterative voting processes.
The distinction between the two is seemingly irreconcilable. On one hand side, preference elicitation assumes truthful but communication limited voters. Iterative voting processes, on the other hand, focus explicitly on manipulative, deceitful voter nature, but make no limiting assumption on the amount of communication required to cheat. Nonetheless, as we will show, the two can co-exist, yielding a more realistic picture of time-extended voting processes with limited communication. To make our intent more clear, below we provide an overview of the two subjects, and highlight our innovation against their background.


\subsection{Preference Elicitation with Sincere Voters}

Traditionally, preference elicitation is performed via multi-stage processes. At each stage the voting center selects one voter and queries for a portion of her preferences. It is assumed that voters respond sincerely, do not know each other's preferences, and are unaware of any intermediate results. The latter, although not completely fool-proof, provides one of the better insurances that voters will not attempt strategic manipulation of the outcome~\cite{RE12}. Under these conditions, the voting center can concentrate on minimizing the number of queries and the amount of information it requires from a voter.

Sequential voting studies the order in which the voting center should present queries to the voters. The strategy here refers only to the voting center, that may or may not alter the outcome, pending on the sequence of queries it selects \cite{conitzer2009hard,xia2011strategic}. However, the assumption is that all voter preferences are known to the voting center. We assume that the voting center does not know anything regarding voter preferences, other than the preferences that have already been submitted.

The communication complexity of preference elicitation for various voting protocols has been analyzed, and upper and lower bounds are available ~\cite{Conitzer05}.
In general, for most voting protocols, in the worst case voters should send their entire set of preferences.
Other theoretical bounds for the computation of necessary winners have been previously addressed ~\cite{betzler2009multivariate,pini2007incompleteness,walsh2007uncertainty}.
At the end of each iteration, it is possible to compute which candidates may still have a chance of winning and which will certainly win. These sets of candidates are known as the set of possible winners and a set of necessary winners respectively~\cite{Konczak05}.
This has been done in various settings, for example, in tournaments \cite{aziz2015possible}.
We adopt this approach and compute the set of possible winners at each stage of the voting process.

A candidate winning set is defined as the set of queries needed in order to determine whether a candidate is a necessary winner. For rules other than the plurality voting, computing this set is NP-hard~\cite{ding_lin_2013}. Following this theorem, heuristics for preference elicitation have been suggested, with the goal of finding the necessary winner using a minimal number of queries. One such heuristic operates under the assumption that each voter holds a predefined decreasing order of the preferences. In an iterative process, the voters are requested to submit their highest preferences; the request is for the rating of a single item from all the voters~\cite{Kalech11}. The major disadvantage of this approach is that it requires the voters to predefine their preferences, which can be inconvenient to the voters. Another heuristic assumes that each voter can be approached only once \cite{Pfeiffer12}. However we assume that each voter can be approached multiple times and that the voters might not predefine their preferences beforehand.

A practical elicitation process that follows these assumptions is proposed for the Borda voting protocol using the minmax regret concept. The output is a definite winner or an approximate winner \cite{lu_boutilier_2011}. Another practical elicitation framework for the Range and Borda voting rules introduces two heuristics for choosing which voter to query regarding which candidates. One heuristic is based on the information gained by the query and the other heuristic tries to maximize the expected score (ES) \cite{naamani2015,naamani-dery_etal_2014, naamani2016reducing}. In this paper we assume that a naive voting center, i.e. a voting center that does not attempt to block strategic behavior, will query the voters according to the ES heuristic.

\subsection{Iterative Voting Processes and Strategic Depth}
Like preference elicitation, iterative voting is also performed via multi-stage processes. At each stage a voter is selected to examine the current election outcome, and is granted the possibility to alter her ballot, after which the election outcome is re-evaluated.

A priori it is unclear whether an iterative voting process will ever stabilize either in strategy (where no voter wishes to change her ballot) or in outcome (where the election outcome no longer changes, even though ballots may). Both converging and cycling voting processes have been demonstrated (see e.g. ~\cite{bcmp_2013,rus1,LR12,MPRJ10,orm_2014_ecai,RW12}). Encouraged by these results, researchers proceeded to study stable points of iterative voting processes (e.g. ~\cite{rolmr_2015_aaai}); investigated voting dynamics, i.e. families of iterative voting strategies, to provide convergence guarantees (e.g. ~\cite{Grandi13,lev2016convergence,MLR14,omprj_2015_aaai,RE12}); and have even expanded the model to include the ability of candidates, rather than voters, to behave strategically ~\cite{brill_conitzer_2015,oepr_2015,porkj_2015}.


Now, it is necessary to separate features of voting dynamics that support convergence, and actual strategic behavior that satisfies these features. Furthermore, even if a convergent strategy is found, can we ensure that it is computationally feasible? This question is not trivial. Although almost all voting rules can be manipulated~\cite{Gib73,Sat75}, it may be difficult to calculate such a manipulation\cite{walsh2011hard}. In fact, in incomplete voting scenarios, both a manipulating strategy and an estimate of the election outcome can be computationally hard (e.g.~\cite{pini2007incompleteness,walsh2007uncertainty}). This computational difficulty is of particular importance for iterative voting, as the information available to a voter at any given stage is incomplete.

One of the major breakthroughs in this direction came from considering a softer form of manipulation by strategic voting: one that will not worsen, but just may improve, the election outcome (e.g.~\cite{eopr_2016,MLR14,RE12}). In a sense, this is a {\em safe} manipulation. It considers all possible situations where the given piece of information holds; and then chooses a ballot augmentation that would improve the outcome in some of those possible scenarios, but damage none. Termed {\em a locally dominant} strategy, this behaviour is myopic in its original design.
However, it is possible to push this idea a bit further, and allow each voter to act non-myopically, assuming even the unlikely case that others will act in her favor~\cite{bn-dorr_2016_m-pref, olprr_MPREF_2016}. Of course this latter, optimistic form of behavior may result in a suboptimal stable point of the iterative voting process.



Against this background, our paper innovates the following fusion: 
\begin{itemize}
\item {\bf Preference elicitation with insincere voters:} Dropping the assumption that voters are always sincere, the elicitation process has to change dramatically as well. In particular, unlike any previous elicitation heuristic, we introduce {\em dual} purpose selection of preference queries. First, as was originally intended, these queries are designed to calculate the election winner as quickly as possible. Second, referring to the core principle of eliciting {\em true} preference, we design these queries to be manipulation resistant.
\item {\bf Iterative voting with partial preferences:} Following the common iterative voting assumption, we assume that intermediate results are available to the voters. However, differently from the common procedure, we investigate a situation where a voter is limited in her communication, and may only answer a given query. 
  \item {\bf  Strategic voting in an iterative preference elicitation process:} We show, however, that in spite of the limited voter-center communication, voters may still manipulate the outcome. To achieve this, we rely on a variant of {\em locally dominant} principle of manipulation, where we introduce additional {\em guidance} in selecting the manipulative vote. Specifically, we require minimization of preference distortion applied to achieve the manipulation effect. Perhaps counter-intuitively, satisfying this additional guidance requirement leads to a polynomial time computable set of manipulations. 
\end{itemize}

Now, it is possible when considering the overall behavior induced by our innovations, to see at first a degree of similarity to {\em sequential voting}. In sequential voting (e.g.~\cite{DE10,conitzer2009hard,xia2011strategic}), rather than casting their ballots simultaneously, voters cast their ballots in a sequence. This way, voters may adapt to preferences of preceding voters. However, a deeper look clearly shows the difference. First, in sequential voting there is no possibility to alter one's vote -- that is the domain of iterative voting processes. Second, the common assumption of sequential voting is that voters reveal their full preference order, as opposed to our explicit limitation on the rate of communicating.

To conclude, the voting process in our paper is a combination of preference elicitation and iterative voting. We borrow from each to describe a more realistic voting interaction. In more detail, the voter is presented with the current election outcome (similar to iterative voting), and is requested to respond to one relative comparison query (similar to some preference elicitation heuristics).  Furthermore, we consider strategic voting behavior, where the voter deviates from her true preferences and submits a response that maximizes her utility. Such interaction is an explicit meld of our two subjects of interest, where each portion necessarily adopts some features of the other. To appreciate this recall that preference elicitation typically assumes sincere voters (which we do not), and iterative voting processes commonly assume the ability to communicate a change of the entire preference order (which we assume to be limited). The fusion is further underlined by our introducing center's awareness of the manipulative nature of voters. In other words, while the main goal of the center remains to determine the winning candidate with a minimal number of queries, the center must now also ensure that outcome reflects true voter preferences, without distortion. To this end, we design a voting center capable of identifying the opportunities of strategic manipulation by voters. This allows the center to balance the need to extract preference information and the need to suppress manipulations. The latter, in our design, takes the form of avoiding certain manipulation-prone queries, potentially stagnating some voters until no other recourse is possible, but query them. Notice that the center and voter population are now integral parts of the same system -- their respective views of the problem are now linked, and all participants are treated as proactive.

\section{The Model} \label{Model}
Our model consists of three elements: a set of candidates $C=\{c_1,...,c_m\}$, a set voters $V=\{v_1,...,v_n\}$, and a voting center. The latter two are active, in the sense of expressing, eliciting and aggregating the opinion of voters about candidates. The opinion of a voter $i$ about various candidates is expressed by a preference order (i.e., a ranking)  $P_i^{true}$ over the set of candidates, i.e., $P_i^{true}=[c_{i_1},...,c_{i_m}]$ where $c_{i_1}$ is $v_i$'s most preferred candidate. We will also denote the relative preference of candidate $c_j$ over $c_k$ by $c_j\succ_{i} c_k$, omitting the subscript where the voter is obvious, or, equivalently by $P_i^{true}(c_j,c_k)$. In addition, we will denote the space of all possible preference orders by $\pSet$. These preferences are an inherent, {\em private} characteristic of each voter, and nothing about $P_i^{true}$ is known to either other voters or the voting center. In fact, not even a possible distribution of preferences within the voter population is a priori known. The voting center, therefore, is tasked with eliciting sufficient portion of these private preferences so as to determine a candidate most agreeable with the given population of voters in the following sense. 

The voting center seeks to implement a voting rule, i.e. a mapping $\winner:\pSet^{n}\rightarrow C$, by which preferences of voters are aggregated and mapped to a single candidate, termed the {\em winner}. In this paper we will assume that the implemented voting rule is Borda~\cite{de1781memoire}. Borda voting rule belongs to a class of rules that associate utility with the relative preference of candidates, and the set of utility values, $\alpha=(\alpha_m\geq \alpha_{m-1}\geq \dots\geq\alpha_1)$, is fixed. In particular, if the top choice of voter $v_i$ is to become the winner, then that voter will receive $\alpha_m$ utility; if the second most preferred candidate of $v_ii$ is to become the winner, then $v_ii$ is presumed to receive $\alpha_{m-1}$; and so on, so that $v_i$ will gain only $\alpha_1$ if her least preferred candidate is declared to be the winner. For Borda, the set of utilities is $\alpha=(m,m-1,\dots,1)$, which allows us to formally define the utility voter $v_i$ extracts from candidate $c_j$ (or {\em score} of candidate $c_j$ from $v_i$) by
$$
\sigma_{i,j}=1+\big|\{c\in C| P_i(c_j,c)\}\big|.
$$
The Borda rule then dictates that, given the joint set of preferences $\vecP=(P_1,\dots,P_n)$ of all voters, the winner is the candidate that yields the maximum total utility across all voters, i.e.:
$$\winner^{Borda}(\vecP)=\arg\max\limits_{c_j\in C}\sum\limits_{v_i\in V}\sigma_{i,j},$$
where a tie between two equally profitable candidates is broken by an a priori set lexicographic order over the candidates. For the brevity of notation, we will omit the ``Borda'' superscript in the remainder of the paper, as we are concentrating on this particular voting rule in our current work. In fact, for the sake of clarity, Table \ref{table:notation_summary} summarizes the notations we use.


Now, in our model, the voting center has no access to the complete, joint preference profile of all voters $\vecP=(P_1,\dots,P_n)$. In fact, the center has no direct access to a complete preference order $P_i$ of any voter. Therefore, the winner calculation is replaced by an interative, approximate elicitation and winner estimation process. We follow~\cite{lu_boutilier_2011,naamani2015} in defining this process. Specifically, at every step of this iterative process a voter-item-item query is generated, that describes a question addressed to a single voter regarding her relative preference among some pair of candidates. Our model does not assume any specific query selection protocol and several possibilities for selecting the voter-item-item queries exist, however we do assume that the center would like to minimise the number of queries that it makes. 




Having been queried via a voter-item-item query $(v_i,c_j,c_k)$, voter $v_i$ can respond either $c_j\succ c_k$, declaring that she prefers $c_j$ to $c_k$, or respond $c_k\succ c_j$ to express the inverse preference. In turn, the voting center collects all query responses into a collection of {\em partial preferences} $\vecQ=(Q_1,\dots,Q_n)\in \famQ^n$, where $Q_i$ is the collection of pair-wise preference expressions from voter $v_i$. Combined with the fact that the voting center seeks to implement the Borda voting rule, $\vecQ$ can be used to estimate the range of possible candidate scores and an estimate of a Borda winner~\cite{Konczak05}. This is done by considering whether a complete joint preference order $\vecP$ exists that could generate (or, more formally, extend) the partial information contained in $CL(\vecQ)$, i.e. all expressed preference comparisons and all preference comparisons inferred by transitivity\footnote{Notice that $CL(\vecQ)$ is simply a composition of $CL(Q_i)$, i.e. transitive closures of individual voter response sets.}.

\begin{definition}{{\bf[Possible and Necessary Winners]} (adapted from~\cite{Konczak05})}\\
Let $P_i\models Q_i$ denote the fact that forall $(c_j,c_k)\in CL(Q_i$ holds $P_i(c_j,c_k)$, and let $Ext(\vecQ)=\{\vecP\in\pSet^n | \forall\ v_i\in V\ P_i\models Q_i\}$, that is we look at all possible complete joint preference profiles that are consistent with the responses recorded by the voting center. Then:
\begin{itemize}
\item $c\in C$ is a {\bf necessary winner} for $\vecQ$ if and only if {\bf forall} $\vecP\in Ext(\vecQ)$ holds $\winner(\vecP)=c$;
\item $c\in C$ is a {\bf possible winner} for $\vecQ$ if there {\bf exists} a $\vecP\in Ext(\vecQ)$ so that $\winner(\vecP)=c$.
\end{itemize}
The Possible Winners for $\vecQ$, $PW=PW(\vecQ)\subseteq C$, is a set that consists of all possible winners for a given set of voter preference responses to queries. 
\end{definition}

Notice that a similar, and somewhat inverse, process is also possible. I.e. we can begin from a set of Possible Winners and ask whether a complete joint preference order $\vecP$ is consistent with it. We will later employ this inverse reasoning in defining domination among possible voter strategies in Section~\ref{sec:model:sub:ldm}. 


The calculation of the PW set drives the voting center to ask further queries, and it does so until the set of possible winners that satisfies all obtained query answers $\vecQ$ is unequivocal, i.e. turns into the necessary winner. At that point, the voting center stops and declares the necessary winner as the winner of the election. 
To drive the voters to disclose their preferences, the voting center, before querying a voter, provides her with the current set of possible winners, PW, similarly to~\cite{conitzer_walsh_xia_2011,RE12}. As voters are selfish, the queried voter is driven to answer a query as this is only manner in which she can effect a change in PW, possibly even in her favor. 
This means that a voter's response to a query is not necessarily dictated by $P_i^{true}$. Rather, a voter may choose to answer a query using an alternative preference order $P_i\neq P_i^{true}$ that forces the voting center to reshape the set of possible winners PW to be more beneficial to the queried voter. In the following Section~\ref{sec:model:sub:ldm} we discuss how such alternative orders compare to each other via the concept of {\em local dominance}, and in Section~\ref{ivc_ldm_theory} we contribute a computational procedure to implement this strategic choice.

However, prior to any such discussion, we must settle the question of why use preference orders at all, rather than just calculating momentary beneficial responses to individual queries. The answer is simple: the voting center calculates the set of possible winners based on the received {\it and the inferred} preferences by transitivity. As a result, the voting center can and will detect logical inconsistency in query answers, if they do not conform (extend) to some complete preference order in $\pSet$. Such inconsistency is implicitly assumed to be punishable, and undesirable by voters and the voting center alike. The need to keep track of query answer consistency over time can be satisfied by maintaining a complete preference order, $P_i$, that can serve as an extension to $Q_i$. Hence, whatever strategic manipulation a voter might want to implement in her query answers, it becomes more conventient to discuss the manipulation in terms of a change that must be introduced into $P_i$. Thus, the overall interaction between the voting center and voters is, thus, summarised by the following loop:
%
\begin{itemize}
\item At the beginning, each voter holds her true set of preferences $P_i=P_{i}^{true}$;
\item As long as the Necessary Winner has not been identified:
\begin{enumerate}[a)]
\item based on the partial preferences of all voters $\mathcal{Q}$, the voting center computes a set of the Possible Winners $PW$.
\item the center selects a voter-item-item query, $\langle v_i\langle c_j,c_k\rangle\rangle$;
\item the voter is provided with the current set of Possible Winners $PW$;
\item the voter decides  whether or not to change her profile from her current profile $P_i$ to a new strategic profile $P_i'$;
\item the voter responds with either $c_j\succ c_k $ or $c_k\succ c_j $ according to her updated profile $P_i$;
\item the center updates the incomplete profile by incorporating $v_i$'s answer and applying transitive closure.
\end{enumerate}
\end{itemize}

Returning to the toy example presented in the Introduction, let us instantiate one step of this loop to clarify the process. Consider three candidates: $c_1,c_2,c_3$ and a voter $v_1$ with the preferences: $c_1\succ c_2$ $\succ c_3$.  
The query is $\langle v_1\langle c_1,c_2\rangle\rangle$  . That is, voter $v_1$ is asked to state her preference between $c_1$ and $c_2$ (step a). The voter is also informed that the possible winners $PW$ are $c_2$ and $c_3$ (step b). Since $c_1$, the candidate most preferred by $v_1$, is not in $PW$, $v_1$ chooses to vote strategically. In order to keep track of all of her strategic moves, the voter updates her profile: $P_1' = c_2\succ c_1$ $\succ c_3$ (step c). The voter submits a response: $c_2\succ c_1$ (step d). The partial profile now known to the center is: $Q_1 = c_2\succ c_1$ (step e).

\begin{table}[H]
\centerline{\begin{tabular}{|l|l|}\hline
{\bf Notation}&{\bf Meaning}\\\hline
$C$&A set of $m$ candidates\\
$V$&A set of $n$ voters\\
$P\in\pSet$& An arbitrary preference profile over $C$ in the space of possible profiles\\
$P_i$& Preference profile of voter $i$\\
$P_{-i}$& Joint preference profile of all voters, but $i$\\
$Q_i$& Reported preferences of voter $i$\\
$\langle v_i\langle c_j,c_k\rangle\rangle$&Query to voter $v_i$ regarding preference between $c_j$ and $c_k$\\
$c\succ c'$&Candidate $c\in C$ is preferred to candidate $c'\in C$\\
$P(c,c')$& $c\succ c'$ holds w.r.t $P$, e.g. $P_i(c_j,c_k)$\\
$\segopen{c}{c'}{P}$& A set of candidates in descending order w.r.t P between candidates $c$ and $c'$\\
$\segment{c}{c'}{P}$& As above, but inclusive of $c$ and $c'$\\
$\infty (-\infty)$& Virtual, most (least) preferred by all candidate\\
$\segleft{\infty}{c'}{Q_i}$& $c'$ and all candidates reported to be preferred to $c'$ by voter $i$\\
$P\downarrow_{T}$& The preference order $P$ limited to a subset of candidates $T\subseteq C$\\
$PW_i$ & The set of possible winners in descending order w.r.t $P_i$\\
$(pw_1^i,...,pw_l^i)$ & Elements of $PW_i$\\
$d_{swap}(P,P')$ & Swap distance between two preference profiles\\
$P\models Q$ & Preference order $P$ consistent with revealed preferences $Q$\\
$\mu(P_i)$ & The set of preference profiles consistent with $Q_i\cup{(c_k,c_j)}$, closest to $P_i$\\
$CL\left(Q_i\cup\{(c_k,c_j)\}\right)$& Transitive closure of all preference in $Q_i\cup{(c_k,c_j)}$\\
\hline
\end{tabular}}
\caption{Notation Summary.\label{table:notation_summary}}
\end{table}

\subsection{Locally Dominant Manipulation}\label{sec:model:sub:ldm}
As we have mentioned, in order to guarantee consistency of their answers, each voter maintains a current preference profile $P_i$ that she uses to answer queries from the voting center. A voter may change this profile, as long as it remains consistent with previous answers, if she deems such a change beneficial for the current and, possibly, future query answering. To capture the benefit of adopting one particular preference order over another we employ the concept of {\em local dominance manipulation model} (see e.g.~\cite{conitzer_walsh_xia_2011,MLR14,RE12}). We formally instantiate local dominance below.

First, notice that because of previous answers given by voters, their ability to adopt additional manipulative changes to their current preference order are limited. Specifically, whatever preference order voter $v_i\in V$ decides to adopt when answering a query $(v_i,c_k,c_j)$, this order will have to be a member of the set $\mathcal{P}_i=\{P\in\pSet|P\models Q_i\}$. Second, the improvement that all voters seek is that of a better final outcome declared by the voting center. Thus, a voter will change her current profile from $P_i$ to $P_i'$ if she recognizes that the change (and query answers that it engenders) {\em possibly} entails a better final outcome, i.e. an outcome that ranks higher with respect to her truthful preference. However, the voter would also want to guarantee that there is {\em no possibility} that the change would result in her being worse off when the final outcome is produced. Let us define these {\em possibilities} more formally.

\begin{definition}{\bf[Possible World and Outcome]}
  Let $PW$ be a set of Possible Winners. A joint preference profile, $P_{-i}$, of all agents in $V\setminus{i}$ is termed {\bf a possible world (from the perspective of $v_i\in V$)}, if it is consistent with the set $PW$. More formally, $P_{-i}$ is a possible world if exists $Q_{-i}$ so that $P_{-i}=Ext(Q_{-i})$ and $PW(Q_{-i},Q_i)=PW$. That is, the use of joint preference profile $P_{-i}$, combined with voter $v_i$'s responses, could have led the voting center to generate $PW$. In turn, {\bf a possible outcome (of $P_i$)}, is the candidate that would be declared the winner by the voting center if it had full access to $P_i$ and some possible world $P_{-i}$.
\end{definition}

Notice that $v_i$ has no access to $P_{-i}^{true}$ or the actual $Q_{-i}$, so the concept of possible worlds and outcomes is speculative. However, this speculation does allow $v_i$ to build a strategic argument, if she {\it implicitly} assumes all other voters to be persistently truthful and that they use $P_{-i}^{true}$ to answer all querries. This is because a possible world $P_{-i}$ may generate different possible outcomes for two different preference orders, $P_i$ and $P_i'$, of agent $i$. Thus, the preference order $P_i$, that voter $v_i$ uses to respond to the voting center queries, can now be used as a strategic manipulation means. More specifically, we can compare the 
relative benefit of the set of possible outcomes of $P_i$ vs the set of possible outcomes of $P_i'$. With that in mind, Local Dominance is a particular form of safe choice of $P_i$ with respect to its effect on the set of possible outcomes. 


\begin{definition}{\bf [Local Dominance]}
  A preference order $P_i'$ is a local dominant over preference order $P_i$
if, in at least one possible world, the possible outcome of $P_i'$ is ranked higher than the possible outcome of $P_i$, and in none of the possible worlds the possible outcome of $P_i'$ is ranked lower than the possible outcome of $P_i$.
\end{definition}

When queried, the voter is requested to submit her preference between two candidates only. However, in order to manipulate, more than a single change in current preferences $P_i$ might be needed. A voter would want to ensure that such a continual distortion does not accumulate into a grotesque misrepresentation of her original preferences $P_i^{true}$. Thus, from all of the possible changes in $P_i$, the voter will seek a change that requires the minimal number of swaps. I.e., the voter actively minimises the change in her current preferences at every step. Formally, to compare two current preference profiles $P_i$ and $P_i'$, we use the {\em swap distance}~\cite{bredereck2016large,kendall_38} defined for two linear orders $P$ and $P'$. The distance counts the number of candidate pairs that are ordered differently by two ballots or linear orders.

Summarising these limitations on a manipulative change in current profile, a voter will change her profile to $P_i'\in \mathcal{P}_i$ only under the following voting manipulation (VM) conditions:
\begin{itemize}
\item {\bf Condition-1:} The new preference profile $P_i'$ is a local dominant profile over $P_i$.
\item {\bf Condition-2:} The new preference profile $P_i'$ has the minimal swap distance out of all possible (consistent with previous $i$'s responses) profiles. \label{vm_mindist}
\end{itemize}
For notational convenience, we will denote the set of all preference profiles that satisfy Condition-2 by $\mu(P_i)$. Formally, $\mu(P_i)=\arg\min_{P\models Q_i\cup\{(c_k,c_j)\}}\kendal(P_i,P)$, where $P\models Q_i$ denotes the fact that the preference order $P$ is consistent with the partial order of reported preferences $Q_i$. In turn, $\kendal(P_i,P)$ denotes the swap distance between two preference profiles.

We now define the scenarios where manipulation can be performed.
For a given set of Possible Winners $PW=\{pw_1,...,pw_l\}$ we define $PW_i$ as the ordered vector of possible winners for voter $v_i$: $PW_i=[pw^i_1,\dots,pw^i_l]$, where the order is w.r.t $P_i$. In particular, for any $1\leq j<k\leq l$, we will have $P_i(pw^i_j,pw^i_k)$, i.e. voter $i$ will prefer $pw^i_j$ to $pw^i_k$. We will omit the superscript where the agent is clear from the content, and simply write $PW_i=[pw_1,\dots,pw_l]$.

To set the conditions for local dominance, we use the following set of {\bf ``Common Givens''}, a set of w.l.o.g. assumptions that we will use throughout the remainder of the paper:
\begin{itemize}
\item The current preference profile of $v_i$, $P_i$;
\item The query is: $\langle v_i\langle c_j,c_k\rangle \rangle$, and according to $P_i$: $c_j\succ c_k$;
\item The ordered (w.r.t. $P_i$) vector of Possible Winners of $v_i$, $PW_i$;
\item The closure of the set of current query responses, $Q_i$;
\item Neither $Q_i(c_j,c_k)$ nor $Q_i(c_k,c_j)$ hold, i.e. there is no committed order among the query's candidates.
\end{itemize}

Now, to describe our algorithms and their theoretical features, we will use an interval-like notation for subsets of candidates. The order within the interval will be that of a preference profile, $P_i$, or the set of previously stated preferences, $Q_i$. In particular, $\segment{c}{c'}{P_i}$ will denote a set of candidates in preference profile $P_i$ between two candidates, $c\succ c'$, {\em inclusive} of the two candidates themselves. At the same time, $\segleft{c}{c'}{P_i}$ will denote the same set, but excluding the candidate $c$. Finally, we will naturally use $\infty$ in this notation, so that, e.g., $\segright{c}{-\infty}{Q_i}$ will denote all candidates that have been reported to be below $c$ and $c$ itself. To maintain the common left-to-right descending order notation, interval boundary points will also appear in descending order. Notice that, in this respect, $+/-\infty$ are used consistently. We will also allow preference orders to be imposed on (or limited to) an interval. Formally, $P\downarrow_{\segment{c}{c'}{P'}}$ will denote a complete order of elements of $\segment{c}{c'}{P'}$ consistent with the preference order $P$.

To instantiate the notation of preference intervals, and preference projections on such intervals, consider the following example. Let $P$ be a preference order over 6(six) candidates so that $c_2\succ c_1\succ c_3\succ c_5\succ c_4\succ c_6$, and $P'=c_4\succ c_6\succ c_1\succ c_5\succ c_3\succ c_2$. Then $[c_1,P,c_4]=\{c_1,c_3,c_4,c_5\}$, while $[c_1,P,c_4)=\{c_1,c_3,c_5\}$, and $\segment{\infty}{c_1}{P'}=\{c_1,c_4,c_6\}$. Finally, the preference order $P''=P\downarrow_{\segment{\infty}{c_1}{P'}}$ will impose the order inherited from $P$ on $\segment{\infty}{c_1}{P'}=\{c_1,c_4,c_6\}$. In particular, according to $P''$ would hold $c_1\succ c_4\succ c_6$.

\subsection{Remarks on model composition}\label{sec:model:}

Some of our modelling decision may raise questions. For instance, the myopic nature of a voter's reaction to queries, as well as the myopic nature of voter's response to queries. One could pose the question of why a strategic, manipulative voter does not consider the response of other voters to her manipulation. In fact, why aren't there any beliefs about the {\em strategic} behaviour of other voters in calculating a manipulation?

In principle, hierarchies of beliefs and cognitive hierarchies, which is what reasoning about strategic reasoning of others may be formally called, have been around for quite some time and across the AI board (e.g., definitely non-exhaustively, ~\cite{monderer_samet_89,Halpern:1990:KCK:79147.79161,boergers_94,Halpern:1985:GML:1625135.1625229,gmytrasiewicz_doshi_2005}), and even of a freshly renewed interest in voting games specifically (e.g.~\cite{chc_2004,cc-gi_2013,DBLP:journals/corr/ElkindGRS17}). However, the practical, human behaviour evidence is that such hierarchies are very shallow\footnote{Many a joke deal with chains of ``I know that he thinks that I know...'', and are in fact based on human inability to handle deep nested beliefs. Comically this is even witnessed in popular culture, e.g. the episode of Friends,''The one where everybody finds out``.} in real-world scenarios~\cite{nagel_95,stahl_wilson_1994}. Furthermore, some of the more recent analysis suggests that a large proportion of people are in fact myopic~\cite{wl-b_aaai_2010}. Rooted in desire to make our model closer to the real world, we take this latter cue and, while assuming some common belief forced by the voting center announcements, adopt myopic strategies by voters. 

There is another, more technical issue that we would like to address: the reason to define manipulations as expressions of some altered version of preferences $P_i\neq P_i^{true}$. Our reasoning is three-fold. First, notice that for any manipulative answer to be effective, there should be no immediate indication of a lie detectable (and potentially punishable) by the voting center. Thus all answers, manipulative or not, should be consistent with each other over time. One of the ways to keep track of this consistency is to maintain a coherent preference order with respect to all answers given, In hindsight, this preference order $P_i$ (potentially different from $P_i^{true}$) can be seen as the order that has generated those answers.

Second, maintaining a coherent preference order allows to measure the general distortion that a manipulative query answer will impose on a voters's ability to express other preferences among candidates. Especially those that are still truthful. There are several ways to treat this distortion over time. One way would be to measure the distortion relative to the original, truthful preference, $P_i^{true}$. However, this would be a measure of how far a voter has {\it already} departed from her original preference in her strive to manipulate the final outcome. The use of such a global reference does not provide the understanding of how much we distort our remaining preferences for future iterations. We believe that a step-by-step distortion, i.e. minimising the distance from the currently used preference order, is a better reflection of future capabilities to manipulatively choose an answer. Notice that it does not mean that $P_i^{true}$ is wholeheartedly abandoned. In fact, as Corollary~\ref{cor:p_true_vs_pw} will show (in the next section), key features of $P_i^{true}$, with regard to those candidate that can win at all, are preserved within $P_i$ throughout all manipulation attempts. Hence our modelling choice in Condition-2.

Third, and perhaps most importantly, the use of a current preference profile $P_i$ allows us to adapt the concept of local strategic dominance. In other words, actually provide some guarantee to the manner in which answering a query effects the final outcome, and strategise over it.

In the following section, we provide our solution to the model. I.e., an efficient computational procedure to find an augmentation to the current preference profile $P_i$ that satisfies all the three aforementioned properties: a) the order must be coherent with all answers given; b) the order must be based on a dominant answer to the current querry; c) the order must allow as much future flexibility as possible, via minimising the momentary swap-distance distortion. 

\section{Interactive Local Dominance}\label{ivc_ldm_theory}
Essentially, our model adapts the concept of Local Dominance to interactive voting scenarios, where voter preferences are only partially known at any given point in time. In this section we provide both the algorithmic and theoretical treatment of the {\bf Interactive Local Dominance (ILD)} concept.

We begin by stating a key feature of ILD, and describe the algorithm to solve it for the particulars of our model. We then proceed with a detailed theoretical analysis of ILD and prove the correctness of our algorithmic solution.

\subsection{Computing ILD response}\label{sec:ild:sub:algo_build}
The following Theorem states that a manipulative response to a query has to maintain the same order of possible winners. Furthermore, for at least one pair of consecutive possible winners the distance between them grows. In turn, the Corollary states that these ordering and distance properties can be consistently traced from the truthful profile throughout all responses of a voter. We defer the proofs of these statements to Section~\ref{subsec:full_theory}, where we will present the encompassing and rigorous theoretical treatment of the guided locally-dominant manipulation.
\begin{theorem}\label{manipulation_features}
Let us assume w.l.o.g. that the enumeration order of possible winners in $PW_i$ is aligned with their order of appearance in $P_i$, that is $pw_a\succ pw_b$ according to $P_i$ if and only if $a<b$.
A preference profile $P_i'$ is a local dominant profile over $P_i$ if and only if the following holds:
\begin{itemize}
\item $pw_\alpha\succ pw_{\alpha+1}$ for all $\alpha\in[1,...,l-1]$ w.r.t. $P_i'$, i.e. the order of possible winners does not change;
\item $|\segment{pw_\alpha}{pw_{\alpha+1}}{P_i'}|\geq |\segment{pw_\alpha}{pw_{\alpha+1}}{P_i}|$ for all $\alpha\in[1,...,l-1]$, i.e. none of the intervals between two consecutive possible winners decreases;
\item Exists $\alpha\in[1,...,l-1]$ so that $|\segment{pw_\alpha}{pw_{\alpha+1}}{P_i'}|\gneq |\segment{pw_\alpha}{pw_{\alpha+1}}{P}|$, i.e. at least one interval between two consecutive possible winner will grow.
\end{itemize}
\end{theorem}

\begin{cor}\label{cor:p_true_vs_pw}
Let $\tau>t$, and $P_i^t$, $P_i^\tau$ are the preference profiles of voter $v_i$ at times $t$ and $\tau$ respectively.
Then the set of possible winners $PW$ at time $\tau$ will be ordered in the same way by $P_i^{true}$ (the truthful preference of $v_i$), $P_i^t$ and $P_i^\tau$. Furthermore, the size of each segment between consecutive possible winners in $PW$ will monotonically grow from $P_i^{true}$ to $P_i^t$ to $P_i^\tau$, and the total size of these segments will grow strictly monotonically.
\end{cor}

Given Theorem~\ref{manipulation_features}, it is easy to confirm whether any of the two given preference profiles, $P,P'$, locally dominates the other. If we are able to compute the set $\mu(P_i)$, then testing its elements would result in finding a feasible manipulation for a voter. To achieve this, we devise a set of characterizations of the set $\mu(P_i)$ in Lemmata~\ref{lemma:general_manipulability}-\ref{lemma:num9}. In turn, this allows us to formulate Algorithm~\ref{algo:meta}, which scans the set of all preference profiles that may belong to $\mu(P_i)$, thus composing $\mu(P_i)$, and confirming that at least one of them locally dominates $P_i$. In particular, as Theorem~\ref{thm:algo_correct} formally confirms, Algorithm~\ref{algo:meta} finds a feasible manipulation, if one exists.

More formally, Algorithm~\ref{algo:meta} operates under the ``Common Givens'' w.l.o.g. assumptions on a query $(c_j,c_k)$, and returns a manipulative preference order, if one exists to satisfy Condition-1 and Condition-2. Otherwise, Algorithm~\ref{algo:meta} keeps the preference order unchanged. Hence, either a new preference order $P_i'$ is returned, where $P_i'(c_k,c_j)$ holds, or the current preference order $P_i$ is kept, where $P_i(c_j,c_k)$. In the former case of the pair the response $c_k\succ c_j$ will be added to $Q_i$, while in the latter -- the query response $c_j\succ c_k$ will be used.
Although formal analysis is necessary to fully appreciate the inner workings of our algorithm (as is done in Section~\ref{subsec:full_theory}), we present its pseudocode before its analysis. This is for the benefit of an engineering-oriented reader whose primary goal is its deployment, rather than further development.

\begin{algorithm}[!ht]
\caption{Voter Manipulation function: Meta-Algorithm \label{algo:meta}}
\begin{algorithmic}[1]
\Require
\Statex ``Common Givens'' w.l.o.g. assumptions
\Statex For query $(c_j,c_k)$ holds $P_i(c_j,c_k)$
\State Set $d^{abs}\leftarrow\infty$, $d^{loc}\leftarrow\infty$, $P^{loc}\leftarrow P_i$
\If {$c_j,c_k$ do {\bf not} satisfy Lemma~\ref{lemma:general_manipulability}} \label{algo:check:lemma:general_manipulability}
\State \Return $P_i$
\EndIf
\State Set $Z\gets \segment{c_j}{c_k}{P_i}$
\For {$z\in Z$ increasing w.r.t $P_i$}
\State $X_{good}\gets \segopen{\infty}{z}{P_i}\setminus\segright{c_j}{-\infty}{Q_i}$\label{algo:build:x_good}
\State $X_{bad}\gets  \segopen{\infty}{z}{P_i}\cap\segopen{c_j}{-\infty}{Q_i}$\label{algo:build:x_bad}
\State $Y_{good} \gets \segright{z}{-\infty}{P_i}\setminus\segleft{\infty}{c_k}{Q_i}$\label{algo:build:y_good}
\State $Y_{bad} \gets \segright{z}{-\infty}{P_i}\cap\segopen{\infty}{c_k}{Q_i}$\label{algo:build:y_bad}
\State Order $X_{good}$, $X_{bad}$, $Y_{good}$ and $Y_{bad}$ by $P_i$\label{algo:build:ensure_order}
\State Denote $P_i'$ the preference order $(X_{good},Y_{bad},c_k,c_j,X_{bad},Y_{good})$
\State $d = \kendal(P_i,P_i')$\label{algo:build:kendal}
\If {$d < d^{abs}$}
\State $d^{abs}=d$
\EndIf
\If {$P_i'$ is LD and $d<d^{loc}$}\label{algo:check:ld}
\State $d^{loc}=d$
\State $P^{loc}=P_i'$
\EndIf
\EndFor
\If {$d^{abs}<d^{loc}$}
\State \Return $P_i$
\Else
\State \Return $P^{loc}$
\EndIf
\end{algorithmic}
\end{algorithm}

\subsection{ILD Response Analysis }\label{subsec:full_theory}

Theorem~\ref{manipulation_features} has given a higher level structure to the set of possible manipulations in our model. It has allowed the reader to build intuition, comprehend our algorithm construction and understand their application examples. Now, it is possible to provide the detailed theoretical treatment, and in this section we will provide complete, rigorous definitions and proofs in support of our algorithmic design.

\begin{proof}[Proof: Theorem~\ref{manipulation_features}]
Let us assume that for some $l$ holds $pw_{l+1}\succ pw_l$. Consider a partial joint profile, $\mathcal{R}_{-i}$, where $\score(pw_l,\mathcal{R}_{-i})=\score(pw_{l+1},\mathcal{R}_{-i})=\eta$ and all other voters have a score of at most $\eta-m$. $\mathcal{R}_{-i}$ is a possible (partial) joint profile given the set $PW_i$. Now, if the voter $i$ submits $P_i$, then $pw_l$ will become the winner. If $P_i'$ is submitted, then $pw_{l+1}$ will win the elections. This contradicts the definition of dominance: $P_i'$ does not dominate $P_i$.

Let us now assume that for some $l$ holds $|\segment{pw_l}{pw_{l+1}}{P_i'}|< |\segment{pw_l}{pw_{l+1}}{P_i}|$. Similar to the previous case, construct a possible (partial) joint profile $\mathcal{R}_{-i}$ so that $\score(pw_{l+1},\mathcal{R}_{-i})-\score(pw_l,\mathcal{R}_{-i})=|\segment{pw_l}{pw_{l+1}}{P_i}|-\eta$. Where $\eta=1$ if $pw_l$ beats $pw_{l+1}$ in tie-breaking and $\eta=2$ otherwise. Furthermore, $\mathcal{R}_{-i}$ can be such that the score of all other candidates is at least $m$ points less than $\score(pw_l,\mathcal{R}_{-i})$. As before, if $P_i$ is submitted by the voter $i$, then $pw_l$ wins the elections, and if $P_i'$ is submitted, then $pw_{l+1}$ wins. Again, this contradicts $P_i'$ dominating $P_i$.

Lastly, assuming that $|\segment{pw_l}{pw_{l+1}}{P_i'}|\geq |\segment{pw_l}{pw_{l+1}}{P_i}|$ for all $l\in[1,...,k-1]$ holds, but there is {\em no} $l\in[1,...,k-1]$ so that $|\segment{pw_l}{pw_{l+1}}{P_i'}|\gneq |\segment{pw_l}{pw_{l+1}}{P}|$. In this case for any $\mathcal{R}=(\mathcal{R}_{-i},P_i)$ holds that $\winner(\mathcal{R})=\winner(\mathcal{R}_{-i},P_i')$. Hence $P_i'$ can not dominate $P_i$.

We conclude that all three conditions are necessary for $P_i'$ to dominate $P_i$. Furthermore, a simple reinspection of the proof quickly leads to the sufficiency of the conditions.
\end{proof}

\begin{theorem}\label{thm:algo_correct}
Assume that ''Common Givens'' w.l.o.g. conditions hold. Algorithm~\ref{algo:meta} works in polynomial time in the number of voters and candidates, and finds a $P_i'$ that satisfies Condition-1 and Condition-2, if such a preference profile exists.
\end{theorem}

Before we prove Theorem~\ref{thm:algo_correct}, i.e. the correctness of Algorithm~\ref{algo:meta}, we provide a set of lemmas that are needed for supporting the proof. All lemmas adopt the ``Common Givens'' w.l.o.g. assumptions mentioned in section 4.1.
Let us examine voter $v_i$'s profile. The preferred order of possible winners according to $v_i$ is: $P_i=[pw_1,pw_2,...,pw_l$]. When $v_i$ is queried for her preference between $c_j$ and $c_k$, her response is: $c_j\succ c_k$. We would like to build a new profile $P_i'$ where voter $v_i$'s response to the same query is: $c_k\succ c_j$. We need $P_i'$ to satisfy conditions Condition-1 and Condition-2, i.e. $P_i'$ should be a local dominant profile over $P_i$ and also have the minimal swap distance to $P_i$ out of all possible profiles.



The only way to create a profile $P_i'$, that is local dominant and has a minimal swap distance, is if in profile $P_i$:
\begin{itemize}
\item $c_j$ is above $pw_1$ and $c_k$ is below $pw_l$: 
$P_i\ \ :\ \ \dots\succ c_j\succ\dots\succ pw_1\succ\dots\succ pw_l\succ\dots\succ c_k\succ\dots$
\item $c_j$ is between $pw_1$ and $pw_l$ and $c_k$ is below $pw_l$:
$P_i\ \ :\ \ \dots\succ pw_1\succ\dots\succ c_j\succ\dots\succ pw_l\succ\dots\succ c_k\succ\dots$
\item $c_j$ is above $pw_1$ and $c_k$ is between $pw_1$ and $pw_l$:
$P_i\ \ :\ \ \dots\succ c_j\succ\dots\succ pw_1\succ\dots\succ c_k\succ\dots\succ pw_l\succ\dots$
\end{itemize}

As an example of the latter case, if $P_i=[c_j,...,pw_1,...,c_k,...pw_l]$ then switching between $c_j$ and $c_k$ by adding $c_j$ to the sequence that is below $pw_1$ and above $pw_l$: $P_i'= [pw_1,...,c_k,c_j ,..., pw_l]$ results in a profile $P_i'$ that is a local dominant with a minimal swap distance, i.e. satisfies both conditions Condition-1 and Condition-2.

Let us consider the alternatives. If in profile $P_i$ both $c_j$ and $c_k$ are either:
\begin{itemize}
\item  below $pw_1$ and above $pw_l$: 
$P_i\ \ :\ \ \dots\succ pw_1\succ\dots\succ c_j\succ\dots\succ c_k\succ\dots\succ pw_l\succ\dots$
\item below $pw_l$:
$P_i\ \ :\ \ \dots\succ pw_1\succ\dots\succ pw_l\succ\dots\succ c_j\succ\dots\succ c_k\succ\dots$
\item above $pw_1$:
$P_i\ \ :\ \ \dots\succ c_j\succ\dots\succ c_k\succ\dots\succ pw_1\succ\dots\succ pw_l\succ\dots$
\end{itemize}

Then for $P_i'$ to be a local dominant profile over $P_i$, the total distance between $pw_1$ and $pw_l$ should increase with respect to the total distance between $pw_1$ and $pw_l$ in profile $P_i$. Therefore, in these cases one must not only switch between $c_j$ and $c_k$ but must also insert at least one candidate between $pw_1$ and $pw_l$ so that the total distance is increased. However, inserting a candidate between $pw_1$ and $pw_l$ results in a profile $P_i'$ that is local dominant but does not have a minimal swap distance. 

Formally, the above can be expressed as:

\begin{lemma}\label{lemma:general_manipulability}
Assume that there is $P_i'\neq P_i$ that satisfies Condition-1 and Condition-2. Then either of the following holds:
\begin{itemize}
\item $P_i(c_j,pw_1)$ and $P_i(pw_l,c_k)$;
\item $P_i(c_j,pw_1)$ and $c_k\in\segment{pw_1}{pw_l}{P_i}$;
\item $P_i(pw_l,c_k)$ and $c_j\in\segment{pw_1}{pw_l}{P_i}$.
\end{itemize}
\end{lemma}

\begin{proof}[Proof: Lemma~\ref{lemma:general_manipulability}]
Let us assume the contrary, i.e. that, in addition to Condition-1 and Condition-2, either of the following holds :
\begin{itemize}
\item $c_j,c_k\in\segment{pw_1}{pw_l}{P_i}$
\item $c_j,c_k\in\segleft{+\infty}{pw_l}{P_i}$
\item $c_j,c_k\in\segright{pw_l}{-\infty}{P_i}$
\end{itemize}

Because Condition-1 holds for $P_i'$, i.e. $P_i'$ locally dominates $P_i$, it follows from Theorem~\ref{manipulation_features} that $$\left|\segment{pw_1}{pw_l}{P_i}\right|<|\segment{pw_1}{pw_l}{P_i'}|.$$
Hence, there is a candidate $c\in C$ so that either $P_i(c,pw_1)$ and $P_i'(pw_1,c)$, or $P_i(pw_l, c)$ and $P_i'(c,pw_l)$. Due to the symmetry of these two cases, let us assume without loss of generality that the former case holds, i.e. $P_i(c,pw_1)$ and $P_i'(pw_1,c)$. Let us assume that $c$ is the highest candidate for which this condition holds with respect to $P_i'$. Formally:
\begin{equation}
\forall\widehat{c}\neq c\ \ P_i(\widehat{c},pw_1) \land P_i'(pw_1,\widehat{c}) \Rightarrow P_i'(c,\widehat{c})\label{lemma:general_manipulability:eqn1}
\end{equation}

Let $c'$ be the candidate immediately above $c$ w.r.t $P_i'$, i.e. $P_i'(c',c)$ and the segment $\segopen{c'}{c}{P_i'}=\emptyset$. Let us show that the candidate pair $(c',c)\neq(c_k,c_j)$, in each of the contrary sub-cases:
\begin{itemize}
\item If $c_j,c_k\in\segment{pw_1}{pw_l}{P_i}$, then $c\neq c_j$ and $c\neq c_k$ since $c\in\segleft{+\infty}{pw_1}{P_i}$.
\item If $c_j,c_k\in\segleft{+\infty}{pw_1}{P_i}$, then $c\neq c_j$, otherwise we obtain contradiction to the Equation~\ref{lemma:general_manipulability:eqn1}, because $P_i'(c_k,c_j)$.
\item If $c_j,c_k\in\segright{pw_l}{-\infty}{P_i}$, then $c\neq c_j$ and $c\neq c_k$ since $c\in\segleft{+\infty}{pw_1}{P_i}$.
\end{itemize}

Furthermore, $c$ and $c'$ are such that $P_i(c,c')$. Otherwise we again obtain contradiction to Equation~\ref{lemma:general_manipulability:eqn1}, since by the choice of $c$ and $c'$ holds that $P_i(c,pw_1)$ and $\left|\segleft{pw_1}{c'}{P_i}\right|\geq 1$ (i.e. $c'$ is either $pw_1$ or below it).

Let us then consider $P_i''$ obtained from $P_i'$ by swapping $c$ and $c'$. It is easy to see that $\kendal(P_i,P_i'')\lneq\kendal(P_i,P_i')$, yet $P_i''\models CL(Q_i\cup{(c_k,c_j)})$. This contradicts the assumption that Condition-2 holds for $P_i'$.

\end{proof}

As before, let us assume that in $P_i$, $c_j\succ c_k$. In $P_i'$ the order of these two candidates is switched so that $c_k\succ c_j$. Let us denote the set of all profiles that have a minimal swap distance from $P_i$ as $\mu(P_i)$.
In order for $P_i' \in\mu(P_i)$, i.e, in order for $P_i'$ to have a minimal swap distance from $P_i$, $c_k$ and $c_j$ need to be ordered directly one after the other,
with no other candidates separating them. Formally:

\begin{lemma}\label{lemma:swap_is_close}
Let $(c_j,c_k)$ be the query, and let there be $c$ so that $P_i'(c_k,c)$ and $P_i'(c,c_j)$, i.e. $(c_k; P_i'; c_j)\neq\emptyset$, then $P_i'\not\in\mu(P_i)$.
\end{lemma}

\begin{proof}[Proof: Lemma~\ref{lemma:swap_is_close}]
Let us have a closer look at the closed interval $[c_k; P_i'; c_j]$. There is a pair of candidates $(c,c')\in[c_k; P_i';c_j]$, so that $P_i(c',c)$ and $(c; P_i', c')=\emptyset$. Because $(c_k; P_i'; c_j)\neq\emptyset$, it holds that $(c,c')\neq(c_k,c_j)$. Let $P_i''$ be a preference order obtained from $P_i'$ by swapping $c$ and $c'$. It is easy to see that $P_i''\models Q_i$ and $\kendal(P_i,P_i'')\lneq\kendal(P_i,P_i')$. I.e. $P_i'\not\in\mu(P_i)$.
\end{proof}


Besides the proximity of $c_j,c_j$, we can also show that certain sets of elements remain in their original order. In particular, the following lemma show that two sub-sets of elements, those with the closest consistent $P_i'\in\mu(P_i)$ places either above $c_k$ or below $c_j$, inherit their relative order from $P_i$.


\begin{lemma}\label{lemma:num3}
Let $P_i'\in\mu(P_i)$, then the following two equations hold
\begin{eqnarray}
P_i'\downarrow_{\segleft{\infty}{c_k}{P_i'}}&=&P_i\downarrow_{\segleft{\infty}{c_k}{P_i'}}\label{lemma:num3:eqn1}\\
P_i'\downarrow_{\segright{c_j}{-\infty}{P_i'}}&=&P_i\downarrow_{\segright{c_j}{-\infty}{P_i'}}\label{lemma:num3:eqn2}
\end{eqnarray}
\end{lemma}

\begin{proof}[Proof: Lemma~\ref{lemma:num3}]
Let us assume that the Equation~\ref{lemma:num3:eqn1} does not hold. Then, there are two candidates, $c,c'$ so that $(c';P_i';c)=\emptyset$, $P_i'(c',c)$ and $P_i(c,c')$. Furthermore, it holds that $P_i'(c,c_k)$. Let us define a new preference order $P_i''$ obtained from $P_i'$ by swapping $c$ and $c'$. Then $\kendal(P_i,P_i'')<\kendal(P_i,P_i')$ and $P_i''\models CL\left(Q_i\cup\{(c_k,c_j)\}\right)$, i.e. $P_i'\not\in\mu(P_i)$, contradicting the lemma's premise.

We obtain the same kind of contradiction by assuming that Equation~\ref{lemma:num3:eqn2} does not hold. Hence the Lemma's conclusion: both Equation~\ref{lemma:num3:eqn1} and~\ref{lemma:num3:eqn2} must hold.
\end{proof}

Furthermore, if we consider two candidates that the original preference order $P_i$ places outside the span between $c_j$ and $c_k$, then they demarcate an upper and a lower candidate intervals that maintain both their order {\em and} composition in $P_i'$.

\begin{lemma}\label{lemma:num7}
Let $P_i'\in\mu(P_i)$, and let $c_j',c_k'\in C$ so that $P_i(c_j',c_j)$ and $P_i(c_k,c_k')$. Then the following equations hold
\begin{eqnarray}
\segleft{\infty}{c_j'}{P_i}&=&\segleft{\infty}{c_j'}{P_i'}\label{lemma:num7:eqn1}\\
\segright{c_k'}{-\infty}{P_i}&=&\segright{c_k'}{-\infty}{P_i'}.\label{lemma:num7:eqn2}
\end{eqnarray}
\end{lemma}

\begin{proof}[Proof: Lemma~\ref{lemma:num7}]
Let us assume that the Equation~\ref{lemma:num7:eqn1} does not hold, in spite of the lemma's premise being true. That is, there exists a candidate $c_j'$ so that $P_i(c_j',c_j)$ and $\segleft{\infty}{c_j'}{P_i}\neq\segleft{\infty}{c_j'}{P_i'}$.

Three possible sub-cases exist in this context:
\begin{enumerate}
\item $\exists c\in C\ s.t.\ P_i(c_j',c)\land P_i'(c,c_j')$
\item $\exists c\in C\ s.t.\ P_i(c,c_j')\land P_i'(c_j',c)$
\item Neither of the above holds.
\end{enumerate}

If $\exists c\in C\ s.t.\ P_i(c_j',c)\land P_i'(c,c_j')$, then it is easy to see that a pair of candidates $(c,c')$ exists so that $\segopen{c}{c'}{P_i'}=\emptyset$, $P_i'(c,c')$, and either $P_i(c',c_j')$ or $c'=c_j'$. Let us then obtain a preference order $P_i''$ from $P_i'$ by swapping $c$ and $c'$. It holds that $\kendal(P_i,P_i'')<\kendal(P_i,P_i')$ and, in addition, $P_i' \models CL\left(Q_i\cup\{(c_k,c_j)\}\right)$. Hence, we contradict the lemma's premise that $P_i'\in\mu(P_i)$.

The sub-case where it holds that $\exists c\in C\ s.t.\ P_i(c,c_j')\land P_i'(c_j',c)$ is similar to the above.

Let us now investigate the third sub-case. It occurs if there is no element that has switched from being above (below) $c_j'$ in $P_i$ to being below (above) $c_j'$ in $P_i'$. In particular the following two sets are equal (as sets):
$$B=\{c\in C| P_i(c,c_j')\}=\{c\in C| P_i'(c,c_j')\}$$
If $\segopen{\infty}{c_j'}{P_i}=\emptyset$, then the assumption of Equation~\ref{lemma:num7:eqn1} not being true can not hold.
If, however, $\segopen{\infty}{c_j'}{P_i}\neq\emptyset$, then $P_i\downarrow_{B}\neq P_i'\downarrow_{B}$. That is, there is a pair of candidates $c,c'\in B$ so that $\segopen{c'}{c}{P_i'}=\emptyset$, $P_i(c,c')$ and $P_i'(c',c)$. Defining an alternative order $P_i''$ obtained from $P_i'$ by swapping $c$ and $c'$, we once again obtain a contradiction to the premise $P_i'\in\mu(P_i)$.

We conclude that Equation~\ref{lemma:num7:eqn1} must hold. Symmetric proof establishes Equation~\ref{lemma:num7:eqn2}.

\end{proof}

Now, as Lemma~\ref{lemma:swap_is_close} showed, $c_j$ and $c_k$ are placed next to each other, when changing the preference order from $P_i$ to $P_i'$. However, to achieve this some other elements may need to be separated. The following lemma shows that this does not occur without need. That is, if two elements were placed next to each other in $P_i$, but not in $P_i'$, then they were separated to accommodate the placement of $c_k$ and $c_j$ between them.

\begin{lemma}\label{lemma:num8}
Let $P_i'\in\mu(P_i)$, and let $a,b\in C$ be two candidates so that $P_i\downarrow_{\{a,b\}}=P_i'\downarrow_{\{a,b\}}$, $\segopen{a}{b}{P_i}=\emptyset$, and $\segopen{a}{b}{P_i'}\neq\emptyset$. Then $a\in\segleft{\infty}{c_k}{P_i'}$ and $b\in\segright{c_j}{-\infty}{P_i'}$.
\end{lemma}

\begin{proof}[Proof: Lemma~\ref{lemma:num8}]
Let us assume that the Lemma's conclusion does not hold. In particular this would mean that $c_j,c_k\not\in\segopen{a}{b}{P_i'}$. On the other hand, $\segopen{a}{b}{P_i'}\neq\emptyset$, so there is a candidate $c\in\segopen{a}{b}{P_i'}$. Because $a$ and $b$ are next to each other in the preference ordering $P_i$, i.e. $\segopen{a}{b}{P_i}=\emptyset$, it holds that either $P_i(c,a)$ or $P_i(b,c)$. Which, in turn, implies that $P_i\downarrow_{\segment{a}{b}{P_i'}}\neq P_i'\downarrow_{\segment{a}{b}{P_i'}}$. Therefore, there is $c'\in\segment{a}{b}{P_i'}$ so that $c'\neq c$ and $\segopen{c'}{c}{P_i'}=\emptyset$, i.e. $c$ and $c'$ are next to each other in the ordering $P_i'$. Furthermore, it must hold that these two elements were switched betwen $P_i$ and $P_i'$, that is $P_i(c,c')$ and $P_i'(c',c)$. Let us define a new preference order $P_i''$ by swapping $c$ and $c'$ in $P_i'$. It would hold that $\kendal(P_i'',P_i)\lneq\kendal(P_i',P_i)$, while $P_i''\models CL\left(Q_i\cup\{(c_k,c_j)\}\right)$, thus contradicting the premise that $P_i'\in\mu(P_i)$.
\end{proof}

One final observation that we will need to prove Theorem~\ref{thm:algo_correct} has to do with the general change in the relative position of elements committed by $Q_i$ to a particular order w.r.t $c_k$ or $c_j$. Lemma~\ref{lemma:num9} shows that among all elements above (correspondingly, below) $c_k$ (correspondingly, $c_j$) only those committed to be ordered after $c_j$ (correspondingly, before $c_k$) will change their relative position when moving from preference order $P_i$ to $P_i'$. All other elements will maintain their order.

\begin{lemma}\label{lemma:num9}
Let us assume that $P_i'\in\mu(P_i)$. Let $c$ be some candidate so that $P_i'(c,c_k)$ and $(c,c_k)\not\in Q_i$. Then the following equality holds:
$$\segleft{\infty}{c}{P_i'}=\segleft{\infty}{c}{P_i}\setminus\segright{c_j}{-\infty}{Q_i}.$$
Symmetrically, let $c$ be some candidate so that $P_i'(c_j,c)$ and $(c_j,c)\not\in Q_i$. Then:
$$\segright{c}{-\infty}{P_i'}=\segright{c}{-\infty}{P_i}\setminus\segleft{\infty}{c_k}{Q_i}.$$
\end{lemma}
\begin{proof}[Proof: Lemma~\ref{lemma:num9}]
  First, notice that if $c\not\in\segopen{c_j}{c_k}{P_i}$, then the lemma is a direct conclusion if Lemma~\ref{lemma:num7}. In more detail, if $P_i(c,c_j)$, then only the premise of the first equation holds. Furthermore, since $P_i\models Q_i$, $\segleft{\infty}{c}{P_i}\cap\segright{c_j}{-\infty}{Q_i}=\emptyset$. Thus, lemma's conclusion requires that $\segleft{\infty}{c}{P_i'}=\segleft{\infty}{c}{P_i}$, which holds due to Lemma~\ref{lemma:num7}. Symmetrically, if $P_i(c_k,c)$, then the premise of the second equation is true, and the conclusion similarly holds according to Lemma~\ref{lemma:num7}. Therefore, in the remainder of our proof of Lemma~\ref{lemma:num9}, we will assume that $c\in\segopen{c_j}{c_k}{P_i}$.

Now, let us assume that there is in fact a candidate $c\in C$ that satisfies the first premise of the lemma, but violates its conclusion. Denote by $X$ the following set:
$$X=\{c\in C| P_i'(c,c_k), (c,c_k)\not\in Q_i, \segleft{+\infty}{c}{P_i'}\neq\segleft{+\infty}{c}{P_i}\setminus\segright{c_j}{-\infty}{Q_i}\}$$
Let $x$ denote the least preferred candidate of $X$ w.r.t the preference order $P_i'$, i.e., for any $x\neq c\in X$ holds that $P_i'(c,x)$.

From Lemma~\ref{lemma:num3} we know that $P_i'\downarrow_{\segleft{+\infty}{x}{P_i'}}=P_i\downarrow_{\segleft{+\infty}{x}{P_i'}}$. Which also means that $\segleft{+\infty}{x}{P_i'}\subseteq\segleft{+\infty}{x}{P_i}$. Furthermore, since $P_i'\models CL\left(Q_i\cup\{(c_k,c_j)\}\right)$, for all $c\in\segright{c_j}{-\infty}{Q_i}$ holds that $P_i'(c_j,c)$ or $c=c_j$. Since $P_i'(x,c_j)$, by the transitivity of $P_i'$ it is also true that $P_i'(x,c)$ for all $c\in\segright{c_j}{-\infty}{Q_i}$. Hence, we obtain $\segleft{+\infty}{x}{P_i'}\cap \segright{c_j}{-\infty}{Q_i}=\emptyset$. In addition, since we have assumed that the lemma's conclusion does not hold, we obtain the following strong subsumption:
\begin{equation}
\segleft{+\infty}{x}{P_i'}\subset\segleft{+\infty}{x}{P_i}\setminus\segright{c_j}{-\infty}{Q_i}\label{lemma:num9:ineq1}
\end{equation}

This means, in particular, that there is a candidate $y\in C$ so that $P_i(y,x)$, $P_i'(x,y)$, and $(c_j,y)\not\in Q_i$. Taking into account Lemma~\ref{lemma:num3}, $P_i$ and $P_i'$ have the following overall structures:
\begin{itemize}
\item $P_i\ \ :\ \ \dots\succ c_j\succ\dots\succ y\succ\dots\succ x\succ\dots\succ c_k\succ\dots$
\item $P_i'\ \ :\ \ \dots\succ x\succ\dots\succ c_k\succ c_j\succ\dots\succ y\succ\dots$
\end{itemize}
Let us denote $A$ the number of candidates between $x$ and $c_k$ with respect to $P_i'$, i.e. $A=|\segleft{x}{c_k}{P_i'}|$, and, correspondingly $B=|\segright{c_j}{y}{P_i'}|$.

Consider now alternative preference orderings $R$ and $R'$, obtained from $P_i'$ by either moving $x$ below $y$ or, alternatively, moving $y$ just below $x$. That is, $P$ and $P'$ have the following structures:
\begin{itemize}
\item $R\ \ :\ \ \dots\succ c_k\succ c_j\succ\dots\succ y\succ x\succ\dots$
\item $R'\ \ :\ \ \dots\succ y\succ x\succ\dots\succ c_k\succ c_j\succ\dots$
\end{itemize}
Furthermore, $P_i'\downarrow_{C\setminus\{x,y\}}=R\downarrow_{C\setminus\{x,y\}}=R'\downarrow_{C\setminus\{x,y\}}$. Let us now denote by $D=\kendal(P_i,P_i')$, and consider $\kendal(P_i,R)$ and $\kendal(P_i,R')$.

It holds that $P_i\downarrow_{\segright{c_j}{y}{P_i'}}=R\downarrow_{\segright{c_j}{y}{P_i'}}$, while $P_i'(x,c)$ for any $c\in\segright{c_j}{y}{P_i'}$, hence $P$ is closer to $P_i$ by $B$ element swaps. At the same time $P_i\downarrow_{\segleft{x}{c_k}{P_i'}}=P_i'\downarrow_{\segleft{x}{c_k}{P_i'}}$, yet $R(c,x)$ for all $c\in\segleft{x}{c_k}{P_i'}$. Similarly the order of $x$ and $y$ is also ''restored'', i.e. it holds that $R(y,x)$, $P_i(y,x)$ and $P_i'(x,y)$. As a result we have $\kendal(P_i,P)=D-B+A-1$. Similarly $\kendal(P_i,R')=D-A+B-1$. Since either $-A+B-1<0$ or $-B+A-1<0$, we have that either $R$ or $R'$ is closer to $P_i$ than $P_i'$. Because no pair of candidates $x,y,c_j,c_k$ is restricted by $Q_i$, we also have that both $R\models CL\left(Q_i\cup\{(c_k,c_j)\}\right)$ and $R'\models CL\left(Q_i\cup\{(c_k,c_j)\}\right)$, therefore violating the assumption of $P_i'\in\mu(P_i)$.
\end{proof}

We are now ready to prove Theorem~\ref{thm:algo_correct} about the correctness of our algorithm.

\begin{proof}[Proof: Theorem~\ref{thm:algo_correct}]
First, let us do away with the question of computational complexity of the Algorithm~\ref{algo:meta}, as the simpler portion of the algorithm's analysis.

Prior to the main loop of the algorithm, a preliminary feasibility of manipulation is run in Line~\ref{algo:check:lemma:general_manipulability}, based on Lemma~\ref{lemma:general_manipulability}. The Lemma includes a finite number of membership checks, each of which runs in time linear in the number of candidates. It does, however, presume that the set of possible winners can be obtained efficiently. Since we use the definition of the PW set from ~\cite{Konczak05,lu_boutilier_2013,naamani-dery_etal_2014}, it can be found efficiently in the number of voters and candidates. Therefore, the overall preliminary check of Line~\ref{algo:check:lemma:general_manipulability} is polynomial in both voter and candidate set sizes.

Once the pre-check is complete the main loop of the algorithm is repeated for every candidate in the worst case. If we show that each loop is polynomial in the size of the problem as well, the overall algorithm's complexity will be obtained.

Lines~\ref{algo:build:x_good}-\ref{algo:build:y_bad} operate on ordered subsets of the candidate set, and each such operation takes at most $2|C|$ basic steps to complete. The arguments of these operations are also polynomial-time constructed. One, $P_i$, is given explicitly as input, and taking a sub-interval of it is linear in $|C|$. The other is obtained, e.g., by a spanning tree traversal of $Q_i$. However, $Q_i$ set is at most quadratic in the number of candidates. Hence, the calculates of sets $X_{good},X_{bad}, Y_{good}$ and $Y_{bad}$ take time, polynomial in the candidate set size. Line~\ref{algo:build:ensure_order} is a sanity check, since the aforementioned subsets of $C$ were obtained from an ordered sequence of candidates, hence linear in $|C|$. Similarly, construction of the new preference order $P_i'$ takes linear time. Calculating the distance $\kendal(P_i,P_i')$ (Line~\ref{algo:build:kendal}) takes at most $|C|^2$ steps, as it is equivalent to running the bubble-sort algorithm. The last non-trivial step, Line~\ref{algo:check:ld}, depends on how efficiently we can confirm the Local Dominance property of $P_i'$ with respect to $P_i$. This confirmation, however, can be performed by using the three conditions of Theorem~\ref{manipulation_features}. As we have already accounted for the calculation of the PW set, each condition of the Theorem takes time linear in $|C|$. We conclude that the main loop runs at most in $O(|C|^3)$. Hence, the overall run time of the algorithm is polynomial in the sizes of $V$ and $C$ sets.

Now, given that we know that the Algorithm~\ref{algo:meta} operates in polynomial time, let us prove that it operates correctly.

Let $P_i'\in\mu(P_i)$, and let us analyze its structure.

Let us assume for the moment, that there are two elements $a_1,b_2\in C$ that satisfy the conditions of Lemma~\ref{lemma:num9}, i.e., $P_i'(a_1,c_k), P_i'(c_j,b_2)$ and $(a_1,c_k)\not\in Q_i$, $(c_j,b_2)\not\in Q_i$. Furthermore, let us assume that $a_1$ is the minimum ($b_2$ is the maximum) such element with respect to $P_i'$. Combining this assumption with Lemma~\ref{lemma:swap_is_close}, $P_i'$ can be broken down into the following structure $P_i'=(F_j,G_k,c_k,c_j,G_j,F_k)$, where the intervals $F_j,F_k,G_j,G_k$ are characterized as follows:
\begin{itemize}
\item $F_j=\segleft{\infty}{a_1}{P_i'}$,
\item $F_k = \segright{b_2}{-\infty}{P_i'}$,
\item $G_k\subset \segopen{\infty}{c_k}{Q_i}$
\item $G_j\subset \segopen{c_j}{-\infty}{Q_i}$
\end{itemize}

Let us now denote $a_2,b_1\in C$, so that $P_i(a_1,a_2), P_i(b_1,b_2)$ and $\segopen{a_1}{a_2}{P_i}=\emptyset$, $\segopen{b_1}{b_2}{P_i}=\emptyset$.  Such $a_2$ and $b_1$ exist, since $P_i(a_1,c_k)$ and $P_i(c_j,b_2)$ due to Lemma~\ref{lemma:num3}. Furthermore, it also entails that $P_i(a_2,c_k)$ and $P_i(b_1,c_j)$.

The following three subcases are possible.
\begin{itemize}
\item Either $\segopen{a_1}{a_2}{P_i'}\neq\emptyset$ or $\segopen{b_1}{b_2}{P_i'}\neq\emptyset$.
\begin{itemize}
\item {\bf Sub-Case A:} $a_2=b_2$
\item {\bf Sub-Case B:} $P_i'(a_2,b_2)$
\end{itemize}
\item {\bf Sub-Case C:} Both $\segopen{a_1}{a_2}{P_i'}=\emptyset$ and $\segopen{b_1}{b_2}{P_i'}=\emptyset$.
\end{itemize}

W.l.g., let us first assume that $\segopen{a_1}{a_2}{P_i'}\neq\emptyset$. Then, according to Lemma~\ref{lemma:num8}, $P_i'(c_j,a_2)$. Combining this with Lemma~\ref{lemma:num3}, we obtain that $P_i(c_j,a_2)$. Thus, we also conclude that $a_2\in\segment{c_j}{c_k}{P_i}$.

\paragraph{Sub-Case A.\\}
If in addition, $a_2=b_2$, then the following holds according to Lemma~\ref{lemma:num9} and seting $z=a_2$ in Algorithm~\ref{algo:meta}:
$$F_j=\segleft{\infty}{a_1}{P_i'}=\segleft{\infty}{a_1}{P_i}\setminus\segright{c_j}{-\infty}{Q_i}=\segopen{\infty}{a_2}{P_i}\setminus\segright{c_j}{-\infty}{Q_i}=X_{good}$$
$$F_k=\segright{a_2}{-\infty}{P_i'}=\segright{a_2}{-\infty}{P_i}\setminus\segleft{\infty}{c_k}{Q_i}=Y_{good}$$
By its definition $G_j\subset \segopen{c_j}{-\infty}{Q_i}$. Furthermore, $G_j\subset\segment{c_j}{a_2}{P_i'}$. Thus, by Lemma~\ref{lemma:num3}, $G_j\subset\segment{c_j}{a_2}{P_i}$. Since $a_2\not\in G_j$, we conclude that $G_j\subset\segopen{\infty}{a_2}{P_i}$. Hence, by setting $z=a_2$ in Algorithm~\ref{algo:meta}, we have:
$$G_j=\segopen{\infty}{a_2}{P_i}\cap\segopen{c_j}{-\infty}{Q_i}=X_{bad}$$
Similarly, $G_k\subset\segopen{\infty}{c_k}{Q_i}$ by definition. Furthermore, $G_k\subset\segment{a_1}{c_k}{P_i'}$. Thus, by Lemma~\ref{lemma:num3}, $G_k\subset\segment{a_1}{c_k}{P_i}$. Since $a_1\not\in G_k$ and $\segopen{a_1}{a_2}{P_i}=\emptyset$, we can conclude that $G_k\subset\segright{a_2}{-\infty}{P_i}$. Letting $z=a_2$ in Algorithm~\ref{algo:meta}, we have:
$$G_k=\segright{z}{-\infty}{P_i}\cap\segopen{\infty}{c_k}{Q_i}=Y_{bad}$$

Thus we have $P_i'=(X_{good},Y_{bad},c_k,c_j,X_{bad},Y_{good})$ during a run of the Algorithm~\ref{algo:meta}, where $z=a_2$. That is, this sub-case of $P_i'\in\mu(P_i)$ will be recovered by the Algorithm~\ref{algo:meta}.
\paragraph{Sub-Case B.\\}
Let us now consider the situation where, rather than $a_2=b_2$, we have $P_i'(a_2,b_2)$. Similar to the case where $a_2=b_2$, we will have that $a_2\in\segment{c_j}{c_k}{P_i}$ and that $X_{good}=F_j$, $Y_{bad}=G_k$ when Algorithm~\ref{algo:meta} constructs a hypothetical manipulative preference profile with $z=a_2$. It remains to show that $X_{bad}$ and $Y_{good}$ combine into the segment $\segopen{c_j}{-\infty}{P_i'}=(G_j,F_k)$, and then conclude that, even if $P_i'(a_2,b_2)$, the preference profile $P_i'$ will be discovered by Algorithm~\ref{algo:meta} for $z=a_2$. To this end, let us have a closer look at segments $\segment{a_2}{b_2}{P_i'}$ and $\segment{a_2}{b_2}{P_i}$.

Let there be $x\in\segment{a_2}{b_2}{P_i}\setminus\segment{a_2}{b_2}{P_i'}$. If $P_i'(c_j,x)$, then, according to Lemma~\ref{lemma:num3}, we obtain a contradiction that $x\in\segment{a_2}{b_2}{P_i'}$. Hence $P_i'(x,c_k)$. On the other hand, it must be that $P_i(a_1,x)$, since $P_i(a_1,a_2)$ and $x\in\segment{a_2}{b_2}{P_i}$. Hence, $x\in G_k\subset \segopen{\infty}{c_k}{Q_i}$. Notice also that, due to Lemma~\ref{lemma:num3}, we have $\segment{a_2}{b_2}{P_i'}\setminus\segment{a_2}{b_2}{P_i}=\emptyset$. As a result, $\segright{a_2}{-\infty}{P_i'}=\segright{a_2}{-\infty}{P_i}\setminus \segopen{\infty}{c_k}{Q_i}=Y_{bad}$, where $Y_{bad}$ is computed for $z=a_2$.


Finally, notice that it is impossible to have $P_i'(b_2,a_2)$, and that the reasoning is symmetric for the case where $\segopen{b_1}{b_2}{P_i'}\neq\emptyset$. Hence, if either $\segopen{a_1}{a_2}{P_i'}\neq\emptyset$ or $\segopen{b_1}{b_2}{P_i'}\neq\emptyset$, then $P_i'$ is discovered by Algorithm~\ref{algo:meta}.
\paragraph{Sub-Case C.\\}
Let us now have a closer look at a $P_i'$ where $\segopen{a_1}{a_2}{P_i'}=\segopen{b_1}{b_2}{P_i'}=\emptyset$.

Denote $d_1,d_2$ a pair of candidates that satisfy conditions\footnote{Notice that such a pair always exists.} of Lemma~\ref{lemma:num9}, and, in addition, that $P_i(a_1,d_1)$ and $(a_1; P_i; d_1)$ is minimal.


Then the reasoning of Sub-Case B above can be repeated, replacing $b_2$ by $d_2$ in its arguments. We conclude that $P_i'$ with $(a_1; P_i'; a_2)=\emptyset$ and $(b_1;P_i';b_2)=\emptyset$ will also be discovered by Algorithm~\ref{algo:meta}:. In other worlds Algorithm~\ref{algo:meta} will discover all elements of $\mu(P_i)$. As the algorithm selects a locally dominant order $P_i'$ among all those that it finds, the final outcome will satisfy both condition Condition-1 and Condition-2.

\end{proof}

\subsection{Careful Voting Center: Securing Against ILD}\label{sec:ild:sub:careful_voting_center}
Our analysis of ILD was directed to find as many safe (in local dominance sense) manipulation opportunities as possible for a voter to adopt. This, however, does not mean that all quries would prompt a manipulation. In fact, Theorem~\ref{manipulation_features} places a clear limitation on the space of manipulable queries. Since the order of possible winners can not be altered by a locally-dominant manipulative change in the current preference profile, all queries regarding relative preference among two possible winners will follow $P_i$. Thus the queried voter will keep her current preference order unchanged.

\begin{cor}
Let us assume that voter $v_i$ at time $t$ has $P_i^t$ as its current preference order, and let $PW$ be the set of possible winners calculated by the voting center at time $t$. Then, forall $c_k,c_j\in PW$, $v_i$ will respond $(c_k,c_j)$ to a voter-item-item query $(v_i,c_k,c_j)$ if $P_i^t(c_k,c_j)$, and $(c_j,c_k)$ if $P_i^t(c_j,c_k)$. In particular, after answering the query $P_i^{t+1}=P_i^t$, i.e. the current preference profile will not change.
\end{cor}

This opens the possibility to secure, at least in part, against manipulative voters affecting the voting center's calculation of the true Borda winner. We can simply avoid, as much as possible, using queries that {\em may} prompt a manipulation attempt by the queried voter.

\begin{definition}
Let us term a query $(v_i,c_k,c_j)$, where both $c_k,c_j\in PW$, a {\bf safe query}. 
A voting center that directs to voters only safe queries, unless no safe query exists, is termed a {\bf careful voting center}.
\end{definition}

\section{Experimental Validation}
We know that manipulations are not summarily avoidable, therefore our experiments do not serve as a manipulation feasibility study. Rather, the experiments provide an
important statistical insight, not addressing the question of “whether”, but “how frequently” manipulations are possible; and, what’s more important, how frequently do
they actually lead to altered election outcome.
In order to evaluate the manipulation impact on the preference elicitation process, we compared manipulative voters to truthful voters in a careful and in a Na{\"i}ve voting center setting.
A careful voting center selects only queries which are not manipulable, when such queries exist. If no such queries exist (since they have been previously used), the voting center stops being careful. A Na{\"i}ve voting center does not consider which queries are manipulable and which are not.

Algorithms that perform preference elicitation in iterations can be found in \cite{lu_boutilier_2011,lu_boutilier_2013,naamani2015,naamani-dery_etal_2014}. In this paper we use the Expected Scored (ES) algorithm found in \cite{naamani2015}. The ES algorithm selects a voter-item-item pair where one of the items is the item with the current maximum score. This algorithm is publicly available whereas some others are used commercially and cannot be tampered with \cite{lu_boutilier_2011,lu_boutilier_2013}. As a baseline we used an algorithm which randomly chooses the next query (denoted as RANDOM). Each algorithm (ES and RANDOM) was studied in three states:
\begin{enumerate}
\item The voters always answer truthfully (ES+T, RANDOM+T).
\item The voters attempt to manipulate (ES+M, RANDOM+M).
\item Manipulative voters with a careful voting center (Careful-ES+M, Careful-RANDOM+M).
\end{enumerate}

Experiments were performed on four real world data sets: the Sushi data set (5000 preference rankings over 10 candidates)~\cite {kamishima_etal_2005}, the T-shirt data set (30 preference rankings over 11 candidates), the Courses-2003 data set (146 preference rankings over 8 candidates) and the Courses-2004 data set (153 preference rankings over 7 candidates). The three latter data sets were taken from the Preflib library \cite{MaWa13a}. The data sets were used to generate responses to elicitation queries, assuming a Borda voting rule. A random set of voter preference profiles ($P$) was sampled with return out of each data set. For each experiment setting, 20 sets of random profiles were evaluated. For each set of profiles, the experiment was conducted 40 times. Thus we reach an amount of 800 experiments for each experiment setting.

The amount of candidates was set to the maximal amount in each data set. The amount of profiles $P$ was first varied on a range of 10, 20,.., 100 voters. Since we suspected that fewer voters may lead to more fragile situations, we also examined scenarios with only 4,5,..,20 voters.

In order to conclude which algorithm performs best over multiple data sets, we followed a robust non-parametric procedure proposed by \cite{Garcia10}. We first used the Friedman Aligned Ranks test in order to reject the null hypothesis that all heuristics perform the same. This was followed by the Bonferroni-Dunn test to find whether one of the heuristics performs significantly better than other heuristics.

\subsection{Measures}
The evaluation focused on the effect of the manipulations on the iterative process. We examined three dependent variables:

\begin{enumerate}
\item Manipulation rate in a careful and a regular voting center \text{-} When a voter has an opportunity to manipulate, she will. We therefore first set out to check how often do the voters actually have an opportunity to manipulate, and whether a careful voting center can reduce the manipulation rate.
\item	 Manipulation impact on the final result \text{-} Even when manipulations occur during the iterative process, it does not mean they impact on the final result, i.e., the chosen winner. We consider a manipulation process successful only if the chosen winner is different than the winner when no manipulations occur.
\item   Manipulations impact on the number of iterations \text{-} How manipulations impact the process length, measured as the number of iterations.
\end{enumerate}

\subsection{Results}\label{sec:exps}
In the following sections, we examine each of the measures in turn. 

\subsubsection{Manipulation rate in a careful and a regular voting center} Across all four data sets, the ratio of queries that have actually been manipulated is very small: the average  manipulation ratio in all experiments is 0.003  for the Sushi, T-shirt and Courses-2003 data sets, and 0.005 for the Courses-2004 data set. The highest result was received for the Courses-2004 data set, with the Random+M algorithm, and 4 voters: 0.0114 of the queries were manipulated.
To illustrate, in Figure \ref{fig:fig1}
we present the results obtained from the four data set. The x-axis presents the number of voters, and the y-axis the ratio of the queries manipulated (i.e., the number of manipulated queries divided by the number of total queries).
The behavior is similar in all data sets.
Unanimously across all data sets, the Friedman test detected a significant difference (with a p-value of at most 0.05) between the four manipulating algorithms: ES+M, RANDOM+M, Careful-ES+M and Careful-RANDOM+M. The Bonferroni-Dunn post-hoc test reveals that Careful-RANDOM+M outperforms the other algorithms, and that ES+M and Careful-ES+M outperform RANDOM+M. Namely, less manipulations occur with ES+M and Careful-ES+M algorithms than in the RANDOM+M algorithm, and even less manipulations occur with the Careful-RANDOM+M algorithm.

At first thought, one might expect Careful-ES+M to exhibit a lower manipulation rate than Careful-RANDOM+M, since RANDOM selects queries at Random. However, the manipulation rate is the lowest for Careful-RANDOM+M, regardless of the number of voters.
This can be explained by the query selection process.
Careful-RANDOM randomly selects a query out of all the safe queries available (i.e., queries that are "safe" cannot be manipulated).
ES algorithms first create a pool of possible queries that are "ES compatible" (according to the ES algorithm). The regular ES algorithm selects one of these queries at Random.
The Careful-ES algorithm has an additional selection requirement - a query that is both ES compatible and safe. If none is found, it will select an ES query that is not safe and is thus exposed to manipulation more often. The set of ES compatible and safe is a subset of the safe queries and thus there are less opportunities for ES to send safe queries.

RANDOM+M is influenced by the number of voters: for less than 20 voters, the curve has a downward slope pattern.
The explanation is simple. For fewer voters, more manipulation possibilities can occur, thus less queries are safe and since RANDOM+M chooses a query at random, it is more prone to manipulation.

\begin{figure}%
\centering
\includegraphics[width=\columnwidth]{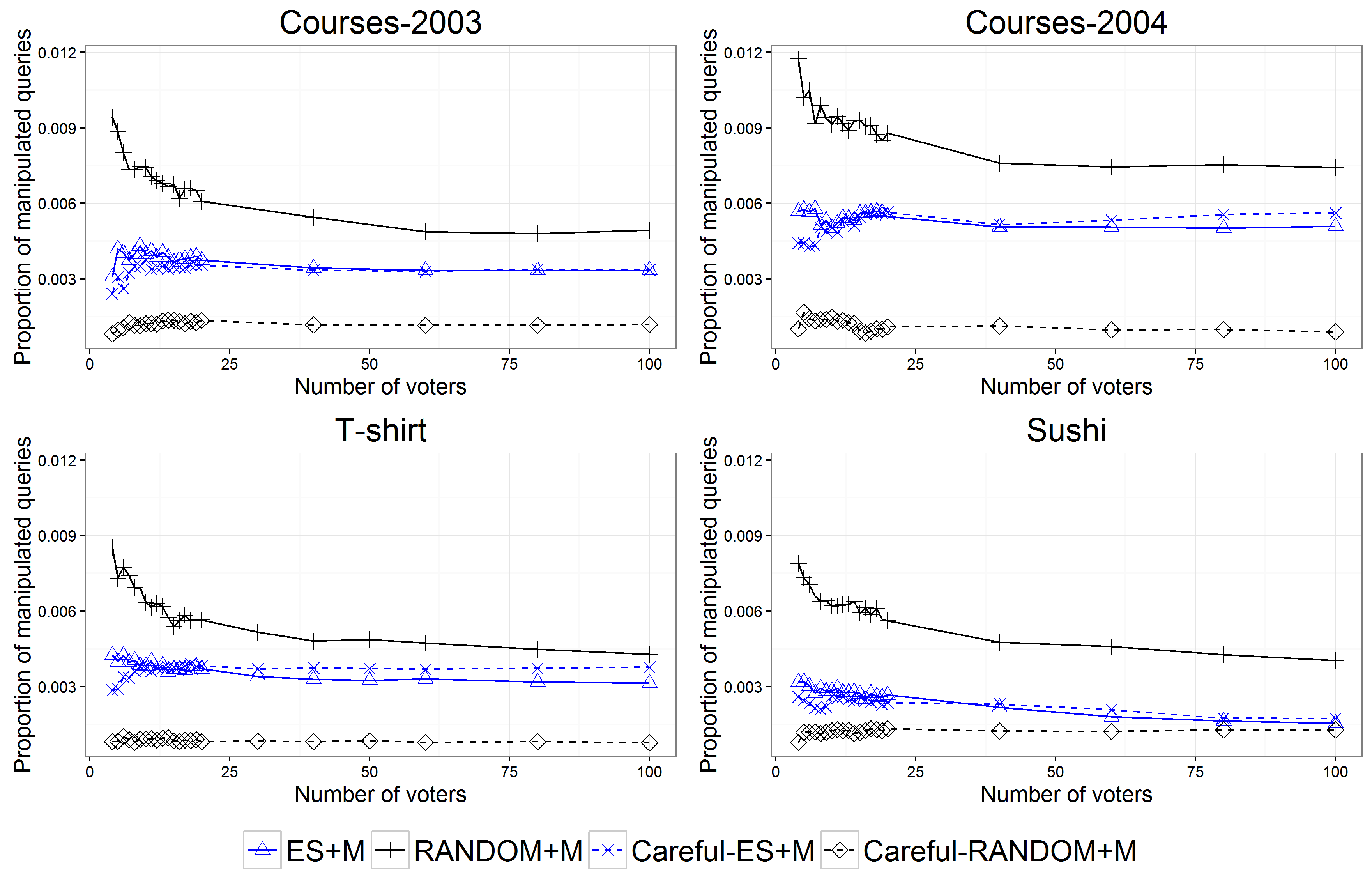}%
\caption{The proportion of manipulated queries. }%
\label{fig:fig1}%
\end{figure}

\subsubsection{Manipulation impact on the final result}
In some cases, there is zero impact on the final result, meaning that although manipulations might occur, they do not alter the final result, i.e., the winning candidate remains the same.
The most severe impact on the final result occurs while using the RANDOM+M algorithm, and even then the proportion of the cases where the final result is changed remains quite low with proportions of: 0.09,0.11,0.13 and 0.16 for the Sushi, Courses-2003,Courses-2004 and T-shirt data sets respectively. Figure \ref{fig:fig3} presents the average proportion of the impact on the final result (i.e., the number of cases where the manipulation changed the final result, out of all cases).
Unanimously across all data sets, the Friedman test detected a significant difference (with a p-value of at most 0.05) between the four manipulating algorithms: ES+M, RANDOM+M, Careful-ES+M and Careful-RANDOM+M.
The Bonferroni-Dunn post-hoc test reveals that Careful-RANDOM+M outperforms RANDOM+M. Namely, using the Careful-Random+M algorithm, the final result is less likely to deviate from the result obtained when all voters vote sincerely.
This conclusion aligns with the previous section - less manipulations occur with the Careful-RANDOM+M, therefore the outcome is less subject to change.

\begin{figure}%
\centering
\includegraphics[width=1\columnwidth]{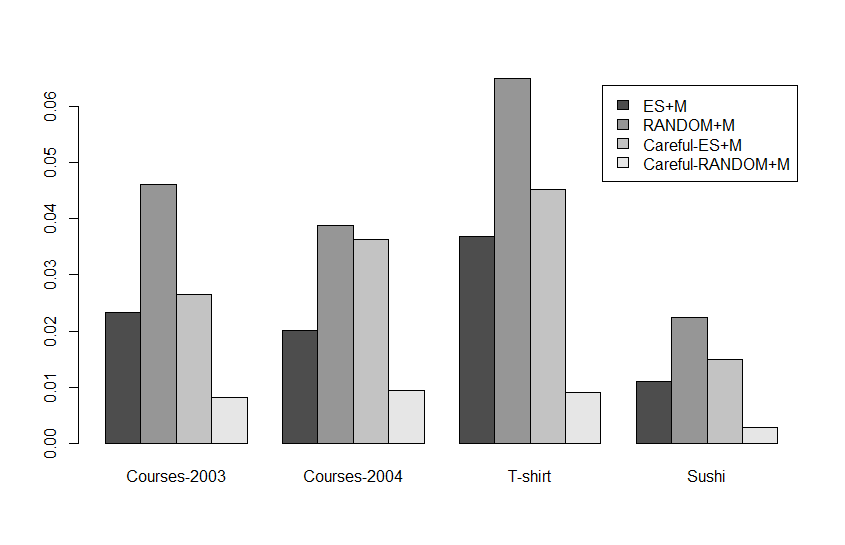}%
\caption{Average proportion of impact on the final result. }%
\label{fig:fig3}%
\end{figure}

\subsubsection{Manipulations impact on the number of iterations}
Figure \ref{fig:fig2} illustrates the average percentage of the data set queried until a necessary winner is found, for a changing number of voters.
Unanimously across all data sets, the Friedman test detected a significant difference (with a p-value of at most 0.05) between the 6 algorithms.
The Bonferroni-Dunn post-hoc test confirmed that ES outperforms all the RANDOM algorithm variations (RANDOM+T, RANDOM+M, Careful-RANDOM+M). Namely, the ES algorithm needs less queries in order to detect a necessary winner.
ES+M outperforms all  RANDOM algorithm variations except in the Courses-04 data set where there is no significant difference between ES+M and RANDOM+T
Careful-ES+M outperforms all  RANDOM algorithm variations except in the Sushi data set where there is no significant difference between Careful-ES+M and RANDOM+T and in the T-shirt data set where there is no significant difference between Careful-ES+M and Careful-RANDOM+M.
There is no significant difference between ES+M and Careful-ES+M. However, there is a significant difference in favor of ES: when no manipulations occur, the result is reached faster. A possible explanation is that manipulative responses cause candidates that would have been removed from the possible winners set in a regular iterative voting to still be considered as possible thus the set of possible winners decreases more slowly.

For the RANDOM algorithms the trend is the opposite: in all data sets but Sushi, Careful-RANDOM+M performs better than the other RANDOM variations. For the Sushi data set, RANDOM+M performs better than the other RANDOM variations.

\begin{figure}%
\centering
\includegraphics[width=1\columnwidth]{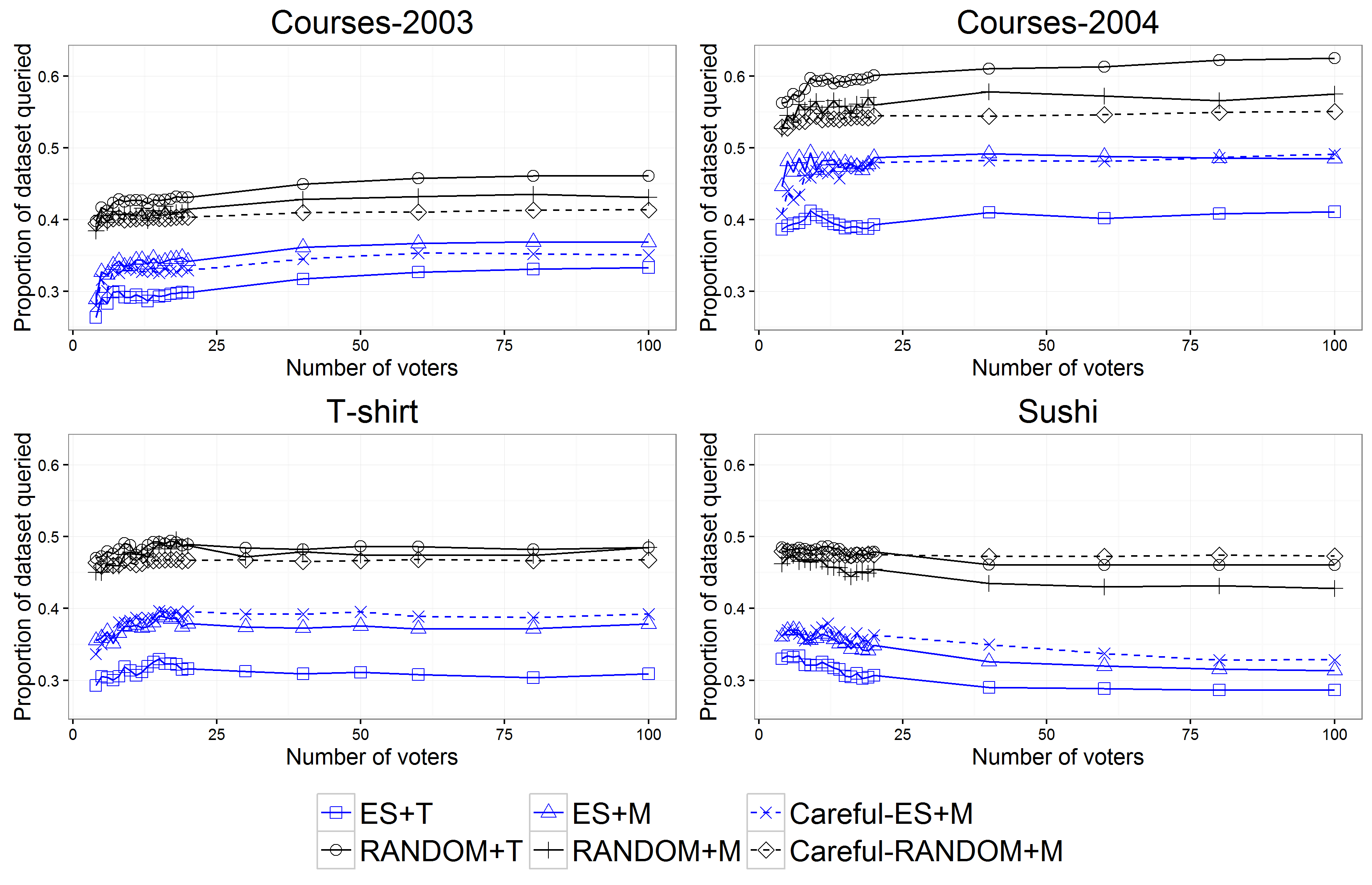}%
\caption{The proportion of the data set queried. }%
\label{fig:fig2}%
\end{figure}

The results indicate that the number of iterations can be reduced when the queries are chosen in a non-Random method - either using the ES algorithm, or using manipulations on the RANDOM algorithms. When some thought is put into choosing the queries, either via a careful voting center or via a non-random algorithm, less queries are used.
However, when the query selection process is not random (e.g. using the ES algorithm) the algorithm itself results in a lower number of iterations, and any additions to the algorithm, such as manipulations, or a careful voting center, only hamper the process.

We can conclude that in practice, manipulations do not pose a big threat on an incremental iterative voting process when the voters submit one preference at a time, since the manipulation rate is low and they rarely modify the outcome. A careful voting center can reduce the manipulation rate, but the tradeoff is that in turn, there is an increase in the amount of queries needed in order to end the iterative process and find a winner. We can therefore state that publishing the list of Possible Winners during a voting process is quite safe from a manipulation perspective.

\section{Conclusions}\label{sec:conclude}
In this work, we have developed and studied a novel combination of two iterative processes found in social choice: iterative preference elicitation using a voting center and the manipulative modification of preferences by voters in Iterative Voting. Traditionally, the design of the former intends to reduce the amount of query requests the voting center sends in order to obtain election outcomes, but assumes voters will reveal their true preference. The latter, on the other hand, presumes that voters may misreport their preferences and vote strategically.
We illustrated how a voter may attempt to manipulate a voting center. We provided a set of algorithms to detect and exploit manipulation opportunities that would drive the voting center to declare an election outcome that is more beneficial to the manipulating voter.
Our manipulation detection algorithms allow us to build a careful voting center that avoids manipulable queries.

Our experiments show that a careful voting center is effective.
Specifically, we show that:
(a) Voters do not have many manipulation opportunities to begin with. Even when we assume that a voter will manipulate whenever she has the opportunity, the ratio of queries that have actually been manipulated is very small.
The careful voting center reduces the manipulations and the manipulation rate is the lowest for Careful-RANDOM+M.
(b) There is a very low impact of manipulations on the final result. Even on the rare occasions that manipulations do occur, they usually do not alter the final result. Again, the careful voting center assists, and the most severe impact on the final result occurs while using the RANDOM+M algorithm.
(c) A careful voting center can reduce the manipulation rate, but the tradeoff is that there is an increase in the amount of queries
needed in order to end the iterative process and find a winner. Another tradeoff is that advanced elicitation schemes such as the one found in~\cite{naamani2015} can reduce the iteration process, but in turn are more prone to manipulations than a Random selection scheme.

Iterative voting processes with preference elicitation are sometimes needed (the reader is referred to the example in the introduction). In such situations, voters may attempt to vote strategically, since the list of possible winners is published or since the voters are able to guess it. We conclude that with careful design, the impact of strategic voting can be diminished.

\bibliography{iter_vote_centre}

\end{document}